\documentclass[11pt]{article}

\usepackage{dsfont}
\usepackage{amsmath}
\usepackage{amsfonts}
\usepackage{ifthen}
\usepackage{natbib}

\usepackage{graphicx}
\usepackage{enumerate}
\usepackage{color}
\usepackage{framed}
\usepackage[pdfpagelabels,pdfpagemode=None,pdfstartview=FitH,bookmarks=false]{hyperref}
\usepackage{amsthm}

\usepackage[margin=1in]{geometry}

\newtheorem{theorem}{Theorem}%
\newtheorem{lemma}{Lemma}%

\newtheorem{definition}{Definition}

\newtheorem{corollary}{Corollary}%

\let\DefaultPar\paragraph
\renewcommand{\paragraph}[1]{\DefaultPar{\textbf{#1}.}}

\DeclareMathAlphabet{\mathpzc}{OT1}{pzc}{m}{it}

\newcommand{\SRef}[1]{\S\ref{#1}}

\newcommand{\AutoAdjust}[3]{\mathchoice{ \left #1 #2  \right #3}{#1 #2 #3}{#1 #2 #3}{#1 #2 #3} }
\newcommand{\Abs}[1]{{\AutoAdjust{\lvert}{#1}{\rvert}}}
\newcommand{\InBrackets}[1]{\AutoAdjust{[}{#1}{]}}
\newcommand{\RangeMN}[2]{\AutoAdjust{\{}{{#1},\ldots,{#2}}{\}}}
\newcommand{\RangeAB}[2]{\AutoAdjust{[}{{#1},{#2}}{]}}

\newcommand{\RangeN}[1]{\InBrackets{#1}}
\newcommand{\Floor}[1]{\AutoAdjust{\lfloor}{{#1}}{\rfloor}}

\newcommand{\Ex}[2][]{\operatorname{\mathbf E}_{#1}\InBrackets{#2}}
\newcommand{\Prx}[2][]{\operatorname{\mathbf{Pr}}_{#1}\InBrackets{#2}}
\newcommand{\PosReals}{\mathds{R}_+}

\newcommand{\PosIntsZ}{\mathbb{N}_{0}}
\newcommand{\Lag}{L}
\newcommand{\sdot}{\;\;} 
\newcommand{\wdot}{\cdot} 
\newcommand{\idot}{\cdot} 

\newcommand{\K}{k}
\newcommand{\N}{n}
\newcommand{\XBox}{x}
\newcommand{\XBoxRV}{{\mathit{X}}}
\newcommand{\OpenProb}{s}
\newcommand{\OpenRV}{{\mathit{S}}}
\newcommand{\Indw}{\ell}
\newcommand{\BrokenRV}{{\mathit{W}}}
\newcommand{\BrokenCDFSign}{{{F}}} %
\newcommand{\BrokenCDFSignW}{\BrokenCDFSign_{\scriptscriptstyle \BrokenRV}} %
\newcommand{\BrokenCDF}[2]{\BrokenCDFSign_{\scriptscriptstyle \BrokenRV_{#1}}(#2)}
\newcommand{\BrokenSel}[2]{{G}_{{#1}}(#2)}
\newcommand{\Threshold}{\theta}
\newcommand{\Distance}{d}

\newcommand{\RV}{{\mathit{V}}}
\newcommand{\CDFRV}{F}
\newcommand{\PDFRV}{f}
\newcommand{\NumRV}{n}
\newcommand{\Stopping}{\tau}
\newcommand{\UVal}{u}
\newcommand{\VVal}{v}

\newcommand{\Mat}{\mathcal{M}}
\newcommand{\IndepSets}{\mathcal{I}}
\newcommand{\IndepSet}{I}
\newcommand{\Rank}[2][]{r^{#1}_{#2}}
\newcommand{\Weight}{w}

\newcommand{\NumAgents}{n}
\newcommand{\NumItems}{m}
\newcommand{\NumUnits}{k}
\newcommand{\MaxDemandFrac}{k}
\newcommand{\BinCount}{c}
\newcommand{\List}{L}

\newcommand{\XAAlloc}{x}
\newcommand{\XAAllocC}{\overline{x}}
\newcommand{\XAAllocA}{\hat{x}}
\newcommand{\XPAlloc}{{\mathit X}}
\newcommand{\Pay}{{\mathit P}}

\newcommand{\XPAllocs}{{\mathit X}}
\newcommand{\Pays}{{\mathit P}}
\newcommand{\Obj}{\operatorname{\textsc{Obj}}}
\newcommand{\PayFactor}{c}
\newcommand{\Valuation}{v}

\newcommand{\Rev}{R}
\newcommand{\ItemSubset}{S}
\newcommand{\Induced}[2]{{#1[#2]}}

\newcommand{\Mech}{{M}}
\newcommand{\Cons}[2]{{#1\langle {#2} \rangle}}
\newcommand{\InducedMech}[2]{[[{#2}]]_{#1}}

\newcommand{\Type}{t}
\newcommand{\Types}{{\mathbf{\Type}}}
\newcommand{\TypeSpace}{T}
\newcommand{\TypeSpaces}{{\mathbf{\TypeSpace}}}
\newcommand{\Dist}{{\mathcal D}}
\newcommand{\Dists}{{\mathbf{{D}}}}
\newcommand{\CDF}{F}

\newcommand{\IP}{IPBR}

\newcommand{\Budget}{B}
\newcommand{\Val}{v}
\newcommand{\CDFB}{\CDF^\Budget}
\newcommand{\Price}{p}
\newcommand{\SRev}{\mathfrak{R}}
\newcommand{\YRev}{\mathit{Y}}
\newcommand{\ZRev}{\mathit{Z}}
\newcommand{\SObj}{u}
\newcommand{\SRevCC}{\widehat{\mathfrak{R}}}
\newcommand{\SRevSC}{\textsc{Rev}}
\newcommand{\TailRev}{r}

\newcommand{\PDF}{f}

\newcommand{\PriceDist}{\mathcal{P}}
\newcommand{\AllocDist}{\mathcal{Q}}

\newcommand{\myref}[2]{\hyperref[#2]{$#1$\ref*{#2}}}

\newcommand{\Mechs}{\mathbb{M}}

\newcommand{\About}[1]{{\textbf{#1.}}}

\newtheorem{remark}{Remark}%

\title{Bayesian Combinatorial Auctions:  Expanding Single Buyer Mechanisms to Many Buyers
\thanks{A preliminary version of this paper was published in the proceedings of the 52nd Annual IEEE Symposium
on Foundations of Computer Science (FOCS 2011).}}

\author{Saeed Alaei
\thanks{Department of Computer Science, University of Maryland, College Park, MD
20742, email: \texttt{saeed@cs.umd.edu}. Part of this work was done when the author was visiting Microsoft
Research, New England. This work was partially supported by the NSF grant CCF-0728839}  }

\date{$ $Date: 2012-06-08 03:38:29 -0400 (Fri, 08 Jun 2012) $ $\protect\footnote{$ $Id: bca.tex 201 2012-06-08 07:38:29Z saeed $ $}}

\begin{document}
\maketitle

\begin{abstract}
For Bayesian combinatorial auctions, we present a general framework for approximately reducing the mechanism
design problem for multiple buyers to single buyer sub-problems. Our framework can be applied to any setting
which roughly satisfies the following assumptions: (i) buyers' types must be distributed independently (not
necessarily identically), (ii) objective function must be linearly separable over the buyers, and (iii)
except for the supply constraints, there should be no other inter-buyer constraints. Our framework is general
in the sense that it makes no explicit assumption about buyers' valuations, type distributions, and single
buyer constraints (e.g., budget, incentive compatibility, etc).

We present two generic multi buyer mechanisms which use single buyer mechanisms as black boxes; if an
$\alpha$-approximate single buyer mechanism can be constructed for each buyer, and if no buyer requires more
than $\frac{1}{\MaxDemandFrac}$ of all units of each item, then our generic multi buyer mechanisms are
$\gamma_\MaxDemandFrac\alpha$-approximation of the optimal multi buyer mechanism, where
$\gamma_\MaxDemandFrac$ is a constant which is at least $1-\frac{1}{\sqrt{\MaxDemandFrac+3}}$. Observe that
$\gamma_\MaxDemandFrac$ is at least $\frac{1}{2}$ (for $\MaxDemandFrac=1$) and approaches $1$ as
$\MaxDemandFrac \to \infty$. As a byproduct of our construction, we present a generalization of prophet
inequalities. Furthermore, as applications of our framework, we present multi buyer mechanisms with improved
approximation factor for several settings from the literature.
\end{abstract}

%

\section{Introduction}
\label{sec:intro}%

The main challenge of stochastic optimization arises from the fact that all instances in the support of the
distribution are relevant for the objective and this support is exponentially big in the size of problem.
This paper aims to address this challenge by providing a general decomposition technique for assignment
problems on independently distributed inputs where the objective is linearly separable over the inputs. The
main challenge faced by such a decomposition approach is that the feasibility constraint of an assignment
problem introduces correlation in the outcome of the optimal solution. In mechanism design problems, such
constraints are typically the supply constraints. For example, when buyers are independent, a revenue
maximizing seller with unlimited supply can decompose the problem over the buyers and optimize for each buyer
independently. However, in the presence of supply constraints, a direct decomposition is not possible. Our
decomposition technique can be roughly described as the following: (i) Construct a mechanism that satisfies
the supply constraints only in expectation (ex-ante); the optimization problem for constructing such a
mechanism can be fully decomposed over the set of buyers. (ii) Convert the mechanism from the previous step
to another mechanism that satisfies the supply constraint at every instance.

%
%
%

We restrict our discussion to Bayesian combinatorial auctions. We are interested in mechanisms that allocate
a set of heterogenous items with limited supply to a set of buyers in order to maximize the expected value of
a certain objective function which is linearly separable over the buyers (e.g., welfare, revenue, etc). The
buyers' types are assumed to be distributed independently according to publicly known priors. We defer the
formal statement of our assumptions to \SRef{sec:prelim}.

The following are the main challenges in designing mechanisms for multiple buyers.
\begin{enumerate}[(I)]
\item \label{dif:n}
The decisions made by the mechanism for different buyers should be coordinated because of supply constraints.

\item \label{dif:1}
The decisions made by the mechanism for each buyer should be optimal (or approximately optimal).
\end{enumerate}
Making coordinated optimal decisions for multiple buyers is challenging as it requires optimizing over the joint type space of
all buyers, the size of which grows exponentially in the number of buyers. The second challenge is usually due to incentive
compatibility (IC) constraints, specially in multi-dimensional settings where these constraints cannot be encoded compactly. In
this paper, we mostly address the first challenge by providing a framework for approximately decomposing the mechanism design
problem for multiple buyers to sub-problems dealing with each buyer individually.


Our framework can be summarized as follows. We start by relaxing the supply constraints, i.e., we consider
the mechanisms for which only the ex-ante expected number of allocated units of each item is no more than the
supply of that item. Note that \emph{``ex-ante''} means that the expectation taken over all possible inputs
(i.e., all possible types of the buyers). We show that the optimal mechanism for the relaxed problem can be
constructed by independently running $n$ single buyer mechanisms, where each single buyer mechanism is
subject to an ex-ante probabilistic supply constraint. In particular, we show that if one can construct an
$\alpha$-approximate mechanism for each single buyer problem, then running these mechanisms simultaneously
and independently yields an $\alpha$-approximate mechanism for the relaxed multiple buyer problem. We then
present two methods for converting the mechanism for the relaxed problem to a mechanism for the original
problem while losing a small constant factor in the approximation. We present two generic multi buyer
mechanisms that use the single buyer mechanisms from the previous step as blackboxes~\footnote{Note that the
single buyer mechanisms can be different for different buyers, e.g., to accommodate different classes of
buyers.}. In the first mechanism, we serve buyers sequentially by running, for each buyer, the corresponding
single buyer mechanism from the previous step. However, we sometimes randomly preclude some of the items from
the early buyers in order to ensure that late buyers get the same chance of being offered with those items;
we ensure that the ex-ante expected probability of preclusion is equalized over all buyers, regardless of the
order in which they are served (i.e., we simultaneously minimize the preclusion probability for all buyers).
In the second mechanism, we run all of the single buyer mechanisms simultaneously and then modify the
outcomes by deallocating some units of the over-allocated items at random while adjusting the payments
respectively; we ensure that the ex-ante probability of deallocation is equalized among all units of each
item and therefore simultaneously minimized for all buyers.

We also introduce a toy problem, \emph{the magician's problem}, in \SRef{sec:magician}, along with a near
optimal solution for it, which is used as the main ingredient of our multi buyer mechanisms. As a byproduct,
we present improved generalized prophet inequalities for maximizing the sum of multiple choices.

As applications of our framework, in \SRef{sec:single}, we present mechanisms with improved approximation
factor for several settings from the literature. For each setting we present a single buyer mechanism that
satisfies the requirements of our framework, and can be plugged in one of our generic multi buyer mechanisms.

\subsection{Related Work}
\label{sec:related}%

In single dimensional settings, the related works form the CS literature are mostly focused on approximating the VCG mechanism
for welfare maximization and/or approximating the Myerson's mechanism~\cite{M81} for revenue maximization (e.g.,
\cite{BR89,BLP06,BH08,HR09,DRY10,CEGMM10,Y11}). Most of them consider mechanisms that have simple implementation and are
computationally efficient. For welfare maximization in single dimensional settings, \cite{HL10} gives a blackbox reduction from
mechanism design to algorithmic design.

In multidimensional setting, for welfare maximization, \cite{HKM11} presents a blackbox reduction from
mechanism design to algorithm design which subsumes the earlier work of \cite{HL10}. For revenue
maximization, \cite{CHMS10} presents several sequential posted pricing mechanisms for various settings with
different types of matroid feasibility constraints. These mechanisms have simple implementation and
approximate the revenue of the optimal mechanism. For unit-demand buyers whose valuations' for the items are
distributed according to product distributions, \cite{CHMS10} present a sequential posted pricing mechanism
that obtains in expectation at least $\frac{1}{6.75}$-fraction of the revenue of the optimal posted pricing
mechanism. In \SRef{sec:unit}, we present an improved sequential posted pricing mechanism for this setting
with an approximation factor of $\frac{1}{2}\gamma_\NumUnits$ in which $\NumUnits$ is the number of units
available of each item, and $\gamma_\NumUnits$ is a constant which is at least
$1-\frac{1}{\sqrt{\NumUnits+3}}$. For combinatorial auctions with additive/correlated valuations with budget
and demand constraints, \cite{BGGM10} presents all-pay $\frac{1}{4}$-approximate BIC mechanisms for revenue
maximization and a similar mechanism for welfare maximization. In \autoref{sec:corr}, we present an improved
mechanism for this setting with an approximation factor of $\gamma_\NumUnits$. Note that $\gamma_k\NumUnits$
is at least $\frac{1}{2}$ and approaches $1$ as $\NumUnits\to \infty$. \cite{BGGM10} also presents sequential
posted pricing mechanisms for the same setting, obtaining $O(1)$ approximation factors. For a similar
setting, in \SRef{sec:budgetm}, we present an improved sequential posted pricing mechanism with an
approximation factor of $(1-\frac{1}{e})\gamma_\NumUnits$. Finally, \cite{CMM11} also considers various
settings with hard budget constraints.

Prophet inequalities have been extensively studied in the past (e.g. \cite{HK92}). Prior to this work, the
best known bound for the generalization to sum of $\NumUnits$ choices was $1-O(\frac{\sqrt{\ln
\NumUnits}}{\sqrt{\NumUnits}})$ by \cite{HKS07}. We improve this to $1-\frac{1}{\sqrt{\NumUnits+3}}$. Note
that the current bound is tight for $\NumUnits=1$, and is useful even for small values of $\NumUnits$.

\section{Preliminaries}
\label{sec:prelim}%
The framework of this paper is presented for combinatorial auctions, but it can be readily applied to Bayesian mechanism design
in other contexts. We begin by defining the model and some notation.

\paragraph{Model}
We consider the problem of selling $\NumItems$ indivisible heterogenous items to $\NumAgents$ buyers where
there are $\NumUnits_j$ units of each item $j \in \RangeN{\NumItems}$. All the relevant private information
of each buyer $i \in \RangeN{\NumAgents}$ is represented by her type $\Type_i \in \TypeSpace_i$ where
$\TypeSpace_i$ is the type space of buyer $i$. Let $\TypeSpaces = \TypeSpace_1 \times \cdots \times
\TypeSpace_n$ be the space of all type profiles. The buyers' type profile $\Types \in \TypeSpaces$ is
distributed according to a publicly known prior $\Dists$. We use $\XPAlloc_{ij}(\Types)$ and $\Pay_i(\Types)$
to denote the random variables\footnote{Note that these random variables are often correlated. Furthermore,
for a deterministic mechanism, these variables take deterministic values as a function of $\Types$.}
respectively for the allocation of item $j$ to buyer $i$ and the payment of buyer $i$, for type profile
$\Types$. For a mechanism $\Mech$, the random variables for allocations and payments are denoted respectively
by $\XPAlloc^\Mech_{ij}(\Types)$ and $\Pay^\Mech_i(\Types)$. We are interested in computing a mechanism that
(approximately) maximizes\footnote{All of the results of this paper can be applied to minimization problems
by simply maximizing the negation of the objective function.} the expected value of a given objective
function $\Obj(\Types,\XPAllocs,\Pays)$ where $\Types$, $\XPAllocs$, and $\Pays$ respectively represent the
types, the allocations, and the payments of all buyers. We are only interested in mechanisms which are within
a given space $\Mechs$ of feasible mechanisms. Formally, we aim to compute a mechanism $\Mech \in \Mechs$
that (approximately) maximizes $\Ex[\Types \sim
\Dists]{\Obj(\Types,\XPAlloc^\Mech(\Types),\Pay^\Mech(\Types))}$.

\paragraph{Assumptions}
We make the following assumptions.
\begin{enumerate}[{(A}1{)}]
\item \About{Independence}
The buyers' types must be distributed independently, i.e., $\Dists = \Dist_1 \times \cdots \times \Dist_n$ where $\Dist_i$ is the
distribution of types for buyer $i$. Note that for a buyer $i$ who has multidimensional types, $\Dist_i$ itself does not need to
be a product distribution.

\item \About{Linear Separability of Objective}
The objective function must be linearly separable over the buyers, i.e., $\Obj(\Type,\XPAllocs,\Pays)=\sum_i
\Obj_i(\Type_i,\XPAlloc_i, \Pay_i)$ where $\Type_i$, $\XPAlloc_i$, and $\Pay_i$ respectively represent the type, the allocations,
and the payment of buyer $i$.

\item \About{Single--Unit Demands}
No buyer should ever need more than \emph{one unit} of each item, i.e., $\XPAlloc_{ij}(\Types) \in \{0,1\}$
for all $\Types \in \TypeSpaces$. This assumption is not necessary and is only to simplify the exposition; it
can be removed as explained in \SRef{sec:multi_unit}.

\item \About{Incentive Compatibility}
\label{asm:ic}%
$\Mechs$ must be restricted to (Bayesian) incentive compatible mechanisms. By direct revelation principle this assumption is
without loss of generality\footnote{It is WLOG, given that we are only interested in mechanisms that have Bayes-Nash
equilibria.},

\item \About{Convexity}
\label{asm:convex}%
$\Mechs$ must be a convex space. In other words, every convex combination of every two mechanisms from
$\Mechs$ must itself be a mechanism in $\Mechs$. A convex combination of two mechanisms $\Mech,\Mech' \in
\Mechs$ is another mechanism $\Mech''$ which simply runs $\Mech$ with probability $\beta$ and runs $\Mech'$
with probability $1-\beta$, for some $\beta \in [0, 1]$. In particular, if $\Mechs$ is restricted to
deterministic mechanisms, it is not convex; however if $\Mechs$ also includes mechanisms that randomize over
deterministic mechanisms, then it is convex \footnote{For an example of a randomized non-convex space of
mechanisms, consider the space of mechanisms where the expected payment of every type must be either less
than $\$2$ or more than $\$4$.}.

\item \About{Decomposability}
\label{asm:decompose}%
The set of constraints specifying $\Mechs$ must be decomposable to supply constraints (i.e., $\sum_i
\XPAlloc_{ij} \le \NumUnits_j$, for each item $j$) and single buyer constraints(e.g., incentive
compatibility, budget, etc). We define this assumption formally as follows. For any mechanism $\Mech$, let
$\InducedMech{i}{\Mech}$ be the single buyer mechanism perceived by buyer $i$, by simulating\footnote{The
single buyer mechanism induced on buyer $i$ can be obtained by simulating all buyers other than $i$ by
drawing a random $\Types_{-i}$ from $\Dists_{-i}$ and running $\Mech$ on buyer $i$ and the $\NumAgents-1$
simulated buyers with types $\Types_{-i}$; note that this is a single-buyer mechanism because the simulated
buyers are just part of the mechanism.} the other buyers according to their respective distributions
$\Dists_{-i}$. Define $\Mechs_i = \{\InducedMech{i}{\Mech}|\Mech \in \Mechs\}$ to be the space of all
feasible single buyer mechanisms for buyer $i$. The decomposability assumption requires that for any
arbitrary mechanism $\Mech$ the following holds: if $\Mech$ satisfies the supply constraints and also
$\InducedMech{i}{\Mech} \in \Mechs_i$ (for all buyers $i$), then it must be that $\Mech \in \Mechs$.
\end{enumerate}

We shall clarify the last assumption by giving an example. Suppose $\Mechs$ is the space of all buyer
specific item pricing mechanisms, then $\Mechs$ satisfies the last assumption. On the other hand, if $\Mechs$
is the space of mechanisms that offer the same set of prices to every buyer, it does not satisfy the
decomposability assumption, because there is an implicit inter-buyer constraint that the same prices should
be offered to different buyers.

Throughout the rest of this paper, we often omit the range of the sums whenever the range is clear from the context (e.g.,
$\sum_i$ means $\sum_{i \in \RangeN{\NumAgents}}$, and $\sum_j$ means $\sum_{j \in \RangeN{\NumItems}}$).

\paragraph{Multi buyer problem}
Formally, the multi buyer problem is to find a mechanism $\Mech$ which is a solution to the following program.
\begin{alignat*}{3}
    \text{maximize}& \qquad& & \sum_i \Ex[\Types\sim \Dists]{\Obj_i(\Type_i,\XPAlloc^\Mech_i(\Types), \Pay^\Mech_i(\Types))} \tag{$OPT$} \label{P} \\
    \text{subject to}& & &\sum_i \XPAlloc^\Mech_{ij}(\Types) \le \NumUnits_j, &
        &\forall \Types \in \TypeSpaces,\forall j \in \RangeN{\NumItems} \\ 
    & & & \InducedMech{i}{\Mech} \in \Mechs_i, &
        & \forall i \in \RangeN{\NumAgents} 
\end{alignat*}
Observe that, in the absence of the first set of constraints, we could optimize the mechanism for each buyer
independently. This observation is the key to our multi to single buyer decomposition, which allows us to
approximately decompose/reduce the multi buyer problem to single buyer problems. A mechanism $\Mech$ is an
\emph{$\alpha$-approximation} of the optimal mechanism if it is a feasible mechanism for the above program
and obtains at least $\alpha$-fraction of the optimal objective value of the program.

\paragraph{Ex ante allocation rule}
For a multi-dimensional mechanism $\Mech$, the \emph{ex ante allocation rule} is a vector $\XAAlloc \in
\RangeAB{0}{1}^{\NumAgents \times \NumItems}$ in which $\XAAlloc_{ij}=\Ex[\Types \sim
\Dists]{\XPAlloc^\Mech_{ij}}$ is the expected probability of allocating a unit of item $j$ to buyer $i$,
where the expectation is taken over all possible type profiles. Note that for any feasible mechanism $\Mech$,
by linearity of expectation, the ex ante allocation rule satisfies $\sum_i \XAAlloc_{ij} \le \NumUnits_j$,
for every item $j$.

\paragraph{Single buyer problem}
The \emph{single buyer problem}, for buyer $i$, is to compute an optimal single buyer mechanism and its
expected objective value, subject to a given upper bound $\XAAllocC_i \in \RangeAB{0}{1}^\NumItems$ on the ex
ante allocation rule; in other words, the single buyer mechanism may not allocate a unit of item $j$ to buyer
$i$ with an expected probability of more than $\XAAllocC_i$, where the expectation is taken over $\Type_i
\sim \Dist_i$. Formally, the single buyer problem is to compute the optimal value of the following program
along with a corresponding solution (i.e., the optimal $\Mech_i$), for a given $\XAAllocC_i$.

\begin{alignat*}{3}
    \text{maximize}& \qquad& & \Ex[\Type_i\sim \Dist_i]{\Obj_i(\Type_i,\XPAlloc^{\Mech_i}_i(\Type_i), \Pay^{\Mech_i}_i(\Type_i))} \tag{$OPT_i$} \label{P.i} \\
    \text{subject to}& & &\Ex[\Type_i \sim \Dist_i]{\XPAlloc^{\Mech_i}_{ij}(\Type_i)} \le \XAAllocC_{ij}, &
        &\forall j \in \RangeN{\NumItems} \\ 
    & & & \Mech_i \in \Mechs_i, &
        &
\end{alignat*}
We typically denote an optimal single buyer mechanism for buyer $i$, subject to a given $\XAAllocC_i$, by
$\Cons{\Mech_i}{\XAAllocC_i}$, and denote its expected objective value (i.e., the optimal value of the above
program as a function of $\XAAllocC_i$) by $\Rev_i(\XAAllocC_i)$. Later, we prove that $\Rev_i(\XAAllocC_i)$,
which we refer to as the \emph{optimal benchmark} for buyer $i$, is a concave function of $\XAAllocC_i$. In
the case of approximation, we say that a single buyer mechanism $\Mech_i$ together with a concave benchmark
$\Rev_i$ provide an $\alpha$-approximation of the optimal single buyer mechanism/optimal benchmark, if the
expected objective value of $\Cons{\Mech_i}{\XAAllocC_i}$ is at least $\alpha\Rev_i(\XAAllocC_i)$ and if
$\Rev_i(\XAAllocC_i)$ is an upper bound on the optimal benchmark, for every $\XAAllocC_i$.

To make the exposition more concrete, consider the following single buyer problem as an example. Suppose
there is only one type of item (i.e., $\NumItems=1$) and the objective is to maximize the expected
revenue\footnote{The optimal multi buyer mechanism for this setting is given by \cite{M81}; yet we consider
this setting to keep the example simple and intuitive.}. Suppose buyer $i$'s valuation is drawn from a
regular distribution with CDF, $\CDF_{i}$. The optimal single buyer mechanism for $i$, subject to
$\XAAllocC_i \in \RangeAB{0}{1}$, is a deterministic mechanism which offers the item at some fixed price,
while ensuring that the probability of sale (i.e., the probability of buyer $i$'s valuation being above the
offered price) is no more than $\XAAllocC_{i}$. In particular, the optimal benchmark $\Rev_i(\XAAllocC_i)$ is
the optimal value of the following convex program as a function of $\XAAllocC_{i}$.
\begin{alignat*}{3}
    \text{maximize}& \qquad& & \XAAlloc_i\CDF_i^{-1}(1-\XAAlloc_i)  \\ 
    \text{subject to}& & & \XAAlloc_i \le \XAAllocC_i  \\
    & & & \XAAlloc_i \in \RangeAB{0}{1}
\end{alignat*}
Furthermore, the optimal single buyer mechanism offers the item at the price $\CDF_i^{-1}(1-\XAAlloc_i)$
where $\XAAlloc_i$ is the optimal assignment for the above convex program. Note that, for a regular
distribution, $\XAAlloc_i\CDF_i^{-1}(1-\XAAlloc_i)$ is concave in $\XAAlloc_i$, so the above program is a
convex program.

\section{Decomposition via Ex ante Allocation Rule}
\label{sec:decompose}%
In this section we present general methods for approximately decomposing/reducing the multi buyer problem to
single buyer problems. Recall that a single buyer problem is to compute the optimal single buyer mechanism
$\Cons{\Mech_i}{\XAAllocC_i}$ and its expected objective value $\Rev_i(\XAAllocC_i)$ (i.e., the optimal
benchmark), subject to an upper bound $\XAAllocC_i$ on the ex ante allocation rule. We present two methods
for constructing an approximately optimal multi buyer mechanism, using $\Mech_i$ and $\Rev_i$ as black box.
Furthermore, we show that if we can only compute an $\alpha$-approximation of the optimal single buyer
mechanism/optimal benchmark for each buyer $i$, then the factor $\alpha$ simply carries over to the
approximation factor of the final multi buyer mechanism.

\paragraph{Multi buyer benchmark}
We start by showing that the optimal value of the following convex program gives an upper bound on the
expected objective value of the optimal multi buyer mechanism.

\begin{alignat*}{3}
    \text{maximize}& \qquad& & \sum_i \Rev_i(\XAAllocC_i) \tag{$\overline{OPT}$} & \qquad & \label{P:OPT'} \\
    \text{subject to}& & &\sum_i \XAAllocC_{ij} \le \NumUnits_j, &
        &\forall j \in \RangeN{\NumItems} \\ 
    & & & \XAAllocC_{ij} \in \RangeAB{0}{1}, &
        &
\end{alignat*}

We first show that the above program is indeed a convex program.
\begin{theorem}
\label{thm:opt_concave}%
The optimal benchmarks $\Rev_i$ are always concave.
\end{theorem}
\begin{proof}
We prove this for an arbitrary buyer $i$. Let $\Mech_i$ and $\Rev_i$ denote the optimal single buyer
mechanism and the optimal benchmark for buyer $i$. To show that $\Rev_i$ is concave, it is enough to show
that for any $\XAAllocC_i,\XAAllocC'_i \in \RangeAB{0}{1}^\NumItems$ and any $\beta \in \RangeAB{0}{1}$, the
following inequality holds.
\begin{align*}
    \Rev_i(\beta \XAAllocC_i+(1-\beta)\XAAllocC'_i) \ge \beta \Rev_i(\XAAllocC_i)+(1-\beta)\Rev_i(\XAAllocC'_i)
\end{align*}
Consider the single buyer mechanism $\Mech''$ that works as follows: $\Mech''$ runs $\Mech_i(\XAAllocC_i)$
with probability $\beta$ and runs $\Mech_i(\XAAllocC'_i)$ with probability $1-\beta$. Note that $\Mechs_i$ is
a convex space (this follows from \myref{A}{asm:convex} and \myref{A}{asm:decompose}), therefore $\Mech'' \in
\Mechs_i$. Observe that by linearly of expectation, the ex ante allocation rule of $\Mech''$ is no more than
$\beta \XAAllocC_i+(1-\beta)\XAAllocC'_i$ and the expected objective value of $\Mech''$ is exactly $\beta
\Rev_i(\XAAllocC_i)+(1-\beta)\Rev_i(\XAAllocC'_i)$. So the expected objective value of the optimal single
buyer mechanism, subject to $\beta \XAAllocC_i+(1-\beta)\XAAllocC'_i$, may only be higher. That implies
$\Rev_i(\beta \XAAllocC_i+(1-\beta)\XAAllocC'_i) \ge \beta \Rev_i(\XAAllocC_i)+(1-\beta)\Rev_i(\XAAllocC'_i)$
which proves the claim.
\end{proof}

\begin{theorem}
\label{thm:decomposition}%
The optimal value of the convex program \eqref{P:OPT'} is an upper bound on the expected objective value of
the optimal multi buyer mechanism.
\end{theorem}
\begin{proof}
Let $\Mech^*$ be an optimal multi buyer mechanism. Let $\XAAlloc^*$ denote the ex ante allocation rule
corresponding to $\Mech^*$, i.e., $\XAAlloc^*_{ij} = \Ex[\Types \sim \Dists]{\XPAlloc^{\Mech^*}_{ij}}$.
Observe that $\XAAlloc^*$ is a feasible assignment for the convex program and yields an objective value of
$\sum_i \Rev_i(\XAAlloc^*_i)$ which is upper bounded by the optimal value of the convex program. So to prove
the theorem it is enough to show that the contribution of each buyer $i$ to the expected objective value of
$\Mech^*$ is upper bounded by $\Rev_i(\XAAlloc^*_i)$. Consider $\Mech^*_i = \InducedMech{i}{\Mech^*}$, i.e,
the single buyer mechanism induced by $\Mech^*$ on buyer $i$. $\Mech^*_i$ can be obtained by simply running
$\Mech^*$ on buyer $i$ and simulating the other $\NumAgents-1$ buyers with random types $\Types_{-i} \sim
\Dists_{-i}$; Observe that $\Mech^*_i$ is a feasible single buyer mechanism subject to $\XAAlloc^*_i$ and
obtains the same expected objective value as $\Mech^*$ from buyer $i$, so the expected objective value of the
optimal single buyer mechanism subject to $\XAAlloc^*_i$ could only be higher.
\end{proof}

\paragraph{Constructing multi buyer mechanisms}
\autoref{thm:decomposition} suggests that by computing an optimal assignment of $\XAAllocC$ for the convex program~\eqref{P:OPT'}
and running the single buyer mechanism $\Cons{\Mech_i}{\XAAllocC_i}$ for each buyer $i$, one might obtain a reasonable multi
buyer mechanism; however such a multi buyer mechanism would only satisfy the supply constraints in expectation; in other words,
there is a good chance that some items are over allocated with a non-zero probability. We present two generic multi buyer
mechanisms for combining the single buyer mechanisms and resolving the conflicts in the allocations in such a way that would
ensure the supply constraints are met at every instance and not just in expectation. In both approaches we first solve the convex
program~\eqref{P:OPT'} to compute the optimal $\XAAllocC$. The high level idea of each mechanism is explained below.
\begin{enumerate}
\item \About{Pre-Rounding}
This mechanism serves the buyers sequentially (arbitrary order); for each buyer $i$, it selects a subset
$\ItemSubset_i$ of available items and runs the single buyer mechanism
$\Cons{\Mech_i}{\Induced{\XAAllocC_i}{\ItemSubset_i}}$, where $\Induced{\XAAllocC_i}{\ItemSubset_i}$ denotes
the vector resulting from $\XAAllocC_i$ by zeroing the entries corresponding to items not in $\ItemSubset_i$.
In particular, this mechanism sometimes precludes some of the available items from early buyers to make them
available to late buyers. We show that if there are at least $\NumUnits$ units of each item, then
$\ItemSubset_i$ includes item $j$ with probability at least $1-\frac{1}{\sqrt{\NumUnits+3}}$, for each buyer
$i$ and each item $j$.
\item \About{Post-Rounding}
This mechanism runs $\Cons{\Mech_i}{\XAAllocC_i}$ for all buyers $i$ simultaneously and independently. It then modifies the
outcomes by deallocating the over allocated items at random in such a way that the probability of deallocation observed by all
buyers are equal, and therefore minimized over all buyers. The payments are adjusted respectively. We show that if there are at
least $\NumUnits$ units of each item, every allocation is preserved with probability $1-\frac{1}{\sqrt{\NumUnits+3}}$ from the
perspective of the corresponding buyer.
\end{enumerate}

We will explain the above mechanisms in more detail in \SRef{sec:expansion} and present some technical assumptions that are
sufficient to ensure that they retain at least $1-\frac{1}{\sqrt{\NumUnits+3}}$ fraction of the expected objective value of each
$\Cons{\Mech_i}{\XAAllocC_i}$.

\paragraph{Approximately optimal single buyer mechanisms}
Throughout the above discussion, we assumed that we can compute the optimal single buyer mechanisms and the
corresponding optimal benchmarks. However, it is likely that we can only compute an approximation of them.
Suppose for each buyer $i$, $\Mech_i$ and $\Rev_i$, instead of being optimal, only provide an
$\alpha$-approximation of the optimal single buyer mechanism/optimal benchmark, and suppose $\Rev_i$ is
concave; then we can still use $\Mech_i$ and $\Rev_i$ in the above construction, but the final approximation
factor will be multiplied by $\alpha$.

\paragraph{Main result}
The following informal theorem summarizes the main result of this paper. The formal statement of this
result can be found in \autoref{thm:pre} and \autoref{thm:post}.

\begin{theorem}[Market Expansion]
\label{thm:expansion}%
If for each buyer $i \in \RangeN{\NumAgents}$, an $\alpha$-approximate single buyer mechanism $\Mech_i$ and a
corresponding concave benchmark $\Rev_i$ can be constructed in polynomial time, then, with some further
assumptions (explained later), a multi buyer mechanism $\Mech \in \Mechs$ can be constructed in polynomial
time by using $\Mech_i$ as building blocks, such that $\Mech$ is $\gamma_\NumUnits\alpha$-approximation of
the the optimal multi buyer mechanism in $\Mechs$, where $\NumUnits = \min_j \NumUnits_j$ and
$\gamma_\NumUnits$ is a constant which is at least $1-\frac{1}{\sqrt{\NumUnits+3}}$.
\end{theorem}

In order to explain the construction of the multi buyer mechanism, we shall first describe the magician's
problem and its solution, which is used in both pre-rounding and post-rounding for equalizing the expected
probabilities of preclusion/deallocation over all buyers.

\section{The Magician's Problem}
\label{sec:magician}%
In this section, we present an abstract online stochastic toy problem and a near-optimal solution for it which provides the main
ingredient for combining single buyer mechanisms to form multi buyer mechanisms; it is also used to prove a generalized prophet
inequality.

\begin{definition}[The Magician's Problem]
\label{def:magician}%
A magician is presented with a series of boxes one by one, in an online fashion. There is a prize hidden in
one of the boxes. The magician has $\K$ magic wands that can be used to open the boxes. If a wand is used on
box $i$, it opens, but with a probability of at most $\XBox_i$, which written on the box, the wand breaks.
The magician wishes to maximize the probability of obtaining the prize, but unfortunately the sequence of
boxes, the written probabilities, and the box in which the prize is hidden are arranged by a villain, and the
magician has no prior information about them (not even the number of the boxes). However, it is guaranteed
that $\sum_i \XBox_i \le \K$, and that the villain has to prepare the sequence of boxes in advance (i.e.,
cannot make any changes once the process has started).
\end{definition}

The magician could fail to open a box either because (a) he might choose to skip the box, or (b) he might run
out of wands before getting to the box. Note that once the magician fixes his strategy, the best strategy for
the villain is to put the prize in the box that has the lowest ex ante probability of being opened, based on
the magician's strategy. Therefore, in order for the magician to obtain the prize with a probability of at
least $\gamma$, he has to devise a strategy that guarantees an ex ante probability of at least $\gamma$ for
opening each box. Notice that allowing the prize to be split among multiple boxes does not affect the
problem. It is easy to show the following strategy ensures an ex ante probability of at least $\frac{1}{4}$
for opening each box: for each box randomize and use a wand with probability $\frac{1}{2}$. But can we do
better?
We present an algorithm parameterized by a probability $\gamma \in \RangeAB{0}{1}$ which guarantees a minimum
ex-ante probability of $\gamma$ for opening each box while trying to minimize the number of wands broken. In
\autoref{thm:magician}, we show that for $\gamma \le 1-\frac{1}{\sqrt{\K+3}}$ this algorithm never requires
more than $\K$ wands.

\begin{samepage}
\begin{definition}[$\gamma$-Conservative Magician] \
\label{def:gcm}%
The magician adaptively computes a sequence of thresholds $\Threshold_1, \Threshold_2,\ldots \in \PosIntsZ$
and makes a decision about each box as follows: let $\BrokenRV_{i}$ denote the number of wands broken prior
to seeing the $i^{th}$ box; the magician makes a decision about box $i$ by comparing $\BrokenRV_{i}$ against
$\Threshold_{i}$; if $\BrokenRV_{i} < \Threshold_{i}$, it opens the box; if $\BrokenRV_{i} > \Threshold_{i}$,
it does not open the box; and if $\BrokenRV_{i} = \Threshold_{i}$, it randomizes and opens the box with some
probability (to be defined). The magician chooses the smallest threshold $\Threshold_{i}$ for which
$\Prx{\BrokenRV_{i} \le \Threshold_{i}} \ge \gamma$ where the probability is computed ex ante (i.e., not
conditioned on past broken wands). Note that $\gamma$ is a parameter that is given. Let
$\BrokenCDF{i}{\Indw}=\Prx{\BrokenRV_{i} \le \Indw}$ denote the ex ante CDF of random variable
$\BrokenRV_{i}$, and let $\OpenRV_{i}$ be the indicator random variable which is $1$ if{f} the magician opens
the box $i$. Formally, the probability with which the magician should open box $i$ condition on
$\BrokenRV_{i}$ is computed as follows.
\begin{align}
     \Prx{\OpenRV_{i}=1|\BrokenRV_{i}} &=
        \begin{cases}
        1                                                               & \BrokenRV_{i} < \Threshold_{i} \\
        (\gamma-\BrokenCDF{i}{\Threshold_{i}-1})/(\BrokenCDF{i}{\Threshold_{i}}-\BrokenCDF{i}{\Threshold_{i}-1}) \quad           & \BrokenRV_{i}= \Threshold_{i} \\
        0                                                               & \BrokenRV_{i} > \Threshold_{i}
        \end{cases} \tag{$\OpenProb$} \label{eq:dp:s} \\
    \Threshold_{i} &=  \min \{\Indw | \BrokenCDF{i}{\Indw} \ge \gamma \}   \tag{$\Threshold$} \label{eq:dp:theta}
\end{align}
Observe that $\Threshold_i$ is in fact computed before seeing box $i$ itself.

Define $\OpenProb_{i}^{\Indw}=\Prx{\OpenRV_{i}=1|\BrokenRV_{i}=\Indw}$; the CDF of $\BrokenRV_{i+1}$ can be
computed from the CDF of $\BrokenRV_{i}$ and $\XBox_i$ as follows (assume $\XBox_i$ is the exact probability
of breaking a wand for box $i$).
\begin{align}
    \BrokenCDF{i+1}{\Indw} &=
        \begin{cases}
        \OpenProb_{i}^{\Indw}\XBox_{i}\BrokenCDF{i}{\Indw-1}+(1-\OpenProb_{i}^{\Indw}\XBox_{i})\BrokenCDF{i}{\Indw}     \quad            & i \ge 1,  \Indw \ge 0 \\
        1                                                                                               & i = 0,  \Indw\ge 0 \\
        0                                                                                               & \text{otherwise.}
        \end{cases} \tag{$\BrokenCDFSignW$} \label{eq:dp:cdf}
\end{align}
Furthermore, if each $\XBox_i$ is just an upper bound on the probability of breaking a wand on box $i$, then
the above definition of $\BrokenCDF{i}{\wdot}$ stochastically dominates the actual CDF of $\BrokenRV_{i}$,
and the magician opens each box with a probability of at least $\gamma$.
\end{definition}
\end{samepage}

In order to prove that a $\gamma$-conservative magician does not fail for a given choice of $\gamma$, we must
show that the thresholds $\Threshold_i$ are no more than $\K-1$. The following theorem states a condition on
$\gamma$ that is sufficient to guarantee that $\Threshold_i \le \K-1$ for all $i$.

\begin{theorem}[$\gamma$-Conservative Magician]
\label{thm:magician}%
For any $\gamma \le 1-\frac{1}{\sqrt{\K+3}}$, a $\gamma$-conservative magician with $\K$ wands opens each box
with an ex ante probability of at least $\gamma$. Furthermore, if $\XBox_i$ is the exact probability (not
just an upper bound) of breaking a wand on box $i$ for each $i$, then each box is opened with an ex ante
probability exactly $\gamma$\footnote{In particular the fact that the probability of the event of breaking a
wand for the $i^{th}$ box is exactly $\XBox_i$, conditioned on any sequence of prior events, implies that
these events are independent for different boxes.}
\end{theorem}
\begin{proof}
See \SRef{sec:magician_analysis}.
\end{proof}

\begin{definition}[$\gamma_\K$]
\label{def:magician_param}%
We define $\gamma_\K$ to be the largest probability such that for any $\K' \ge \K$ and any instance of the
magician's problem with $\K'$ wands, the thresholds computed by a $\gamma_\K$-conservative magician are less
than $\K'$. In other words, $\gamma_\K$ is the optimal choice of $\gamma$ which works for all instances with
$\K' \ge \K$ wands. By \autoref{thm:magician}, we know that $\gamma_\K$ must be\footnote{Because for any $\K'
\ge \K$ obviously $1-\frac{1}{\sqrt{\K+3}} \le 1-\frac{1}{\sqrt{\K'+3}}$.} at least
$1-\frac{1}{\sqrt{\K+3}}$.
\end{definition}

Observe that $\gamma_\K$ is a non-decreasing function in $\K$ which is at least $\frac{1}{2}$ (when $\K=1$) and approaches $1$ as
$\K \to \infty$. The next theorem shows that the lower bound of $1-\frac{1}{\sqrt{\K+3}}$ on $\gamma_\K$ cannot be far from the
optimal.

\begin{theorem}[Hardness of Magician's Problem]
\label{thm:magician_hardness}%
For any $\epsilon > 0$, it is not possible to guarantee an ex ante probability of $1-\frac{\K^\K}{e^\K
\K!}+\epsilon$ for opening each box (i.e., no magician can guarantee it). Note that $1-\frac{\K^\K}{e^\K \K!}
\approx 1-\frac{1}{\sqrt{2\pi \K}}$ by Stirling's approximation.
\end{theorem}
\begin{proof}
See \SRef{proof:thm:magician_hardness}.
\end{proof}

\paragraph{Prophet Inequalities}
We prove a generalization of prophet inequalities by a direct reduction to the magician's problem.

\begin{definition}[$\K$-Choice Sum]
A sequence of $\NumRV$ non-negative random numbers $\RV_1, \ldots, \RV_\NumRV$ are drawn from arbitrary
distributions $\CDFRV_1,\ldots, \CDFRV_\NumRV$ one by one in an arbitrary order. A gambler observes the
process and may select $\K$ of the random numbers, with the goal of maximizing the sum of the selected ones;
a random number may only be selected at the time it is drawn, and it cannot be unselected later. The gambler
knows all the distributions in advance, and observes from which distribution the current number is drawn, but
not the order in which the future numbers are drawn. On the other hand, a prophet knows all the actual draws
in advance, so he chooses the $\K$ highest draws. We assume that the order in which the random numbers are
drawn is fixed in advance (i.e., may not change based on the decisions of the gambler).
\end{definition}

\cite{HKS07} proved that there is a strategy for the gambler that guarantees in expectation at least
$1-O(\frac{\sqrt{\ln \K}}{\sqrt{\K}})$ fraction of the payoff of the prophet, using a non-decreasing sequence
of $\K$ stopping rules (thresholds) \footnote{A gambler with stopping rules $\Stopping_1, \ldots,
\Stopping_k$ works as follows. Upon seeing $\RV_i$, he selects it iff $\RV_i \ge \Stopping_{j+1}$ where $j$
is the number of random draws selected so far.}. Next, we construct a gambler that obtains in expectation at
least $\gamma_\K$ fraction of the prophet's payoff, using a $\gamma_\K$-conservative magician as a black box.
Note that $\gamma_\K \ge 1-\frac{1}{\sqrt{\K+3}}$. This gambler uses only a single threshold. However, he may
skip some of the random variables at random.

\begin{theorem}[Prophet Inequalities -- $\K$-Choice Sum]
\label{thm:gambler}%
The following strategy ensures that the gambler obtains at least $\gamma_\K$ fraction of payoff of the
prophet in expectation. \footnote{To simplify the exposition we assume that the distributions do not have
point masses. The result holds with slight modifications if we allow point masses.}
\begin{samepage}
\begin{itemize}
\item
Find a threshold $\Stopping$ such that $\sum_i \Prx{\RV_i > \Stopping} = \K$ (e.g., by doing a binary search on $\Stopping$).
\item
Use a $\gamma_\K$-conservative magician with $\K$ wands. Upon seeing each $\RV_i$, create a box and write
$\XBox_i = \Prx{\RV_i > \Stopping}$ on it and present it to the magician. If the magician chooses to open the
box and also $\RV_i
> \Stopping$, then select $\RV_i$ and break the magician's wand, otherwise skip $\RV_i$.
\end{itemize}
\end{samepage}
\end{theorem}
\begin{proof}
First, we compute an upper bound on the expected payoff of the prophet. Let $\XBox_i$ be the ex ante
probability (i.e., before any random number is drawn) that the prophet chooses $\RV_i$ (i.e. the probability
that $\RV_i$ is among the $\K$ highest draws). Let $\UVal_i(\XBox_i)$ denote the maximum possible
contribution of the random variable $\RV_i$ to the expected payoff of the prophet if $\RV_i$ is selected with
an ex ante probability $\XBox_i$. Note that $\UVal_i(\XBox_i)$ is equal to the expected value of $\RV_i$
conditioned on being above the $1-\XBox_i$ quantile, multiplied by the probability of $\RV_i$ being above
that quantile. Assuming $\CDFRV_i(\wdot)$ and $\PDFRV_i(\wdot)$ denote the CDF and PDF of $\RV_i$, we can
write $\UVal_i(\XBox_i) = \int_{\CDFRV_i^{-1}(1-\XBox_i)}^\infty \VVal \PDFRV_i(\VVal) d\VVal$. By changing
the integration variable and applying the chain rule we get $\UVal_i(\XBox_i) = \int_0^{\XBox_i}
\CDFRV_i^{-1}(1-\XBox) d\XBox$. Observe that $\frac{d}{d\XBox_i} \UVal_i(\XBox_i) = \CDFRV_i^{-1}(1-\XBox_i)$
is a non-increasing function, so $\UVal_i(\XBox_i)$ is a concave function. Furthermore, $\sum_i \XBox_i \le
\K$ because the prophet cannot choose more than $\K$ random draws. So the optimal value of the following
convex program is an upper bound on the payoff of the prophet.

\begin{alignat*}{3}
    \text{maximize} &   &\qquad & \sum_i \UVal_i(\XBox_i) \\
    \text{subject to}&  &       & \sum_i  \XBox_i \le \K  \tag{$\Stopping$} \\
                    &   &       & \XBox_i \ge 0, & & \forall i \in \RangeN{\NumRV} \tag{$\mu_i$}
\end{alignat*}

Define the Lagrangian for the above convex program as
\[\Lag(\XBox,\Stopping,\mu)=-\sum_i \UVal_i(\XBox_i)+\Stopping\sdot\left(\sum_i \XBox_i-\K\right)-\sum_i \mu_i \XBox_i.\]
By KKT stationarity condition, at the optimal assignment, it must be $\frac{\partial}{\partial \XBox_i}
\Lag(q,\Stopping,\mu) = 0$. On the other hand, $\frac{\partial}{\partial \XBox_i}
\Lag(q,\Stopping,\mu)=-\CDFRV_i^{-1}(1-\XBox_i)+\Stopping-\mu_i$. Assuming that $\XBox_i
> 0$, by complementary slackness $\mu_i = 0$, which then implies that $\XBox_i=1-\CDFRV_i(\Stopping)$, so $\XBox_i=\Prx{\RV_i > \Stopping}$.
Furthermore, it is easy to show that the first constraint must be tight, which implies that $\sum_i \Prx{\RV_i > \Stopping}=\K$.
Observe that the contribution of each $\RV_i$ to the objective value of the convex program is exactly $\Ex{\RV_i|\RV_i
> \Stopping} \Prx{\RV_i > \Stopping}$. By using a $\gamma_\K$-conservative magician we can ensure that each box is opened with
probability at least $\gamma_\K$ which implies the contribution of each $\RV_i$ to the expected payoff of the gambler is
$\Ex{\RV_i|\RV_i > \Stopping} \Prx{\RV_i > \Stopping} \gamma_\K$ which proves that the expected payoff of the gambler is at least
$\gamma_\K$ fraction of optimal objective value of the convex program, which was itself and upper bound on the expected payoff of
the prophet.
\end{proof}

\section{Generic Multi Buyer Mechanisms}
\label{sec:expansion}%

In this section, we present a formal description of the two generic multi buyer mechanisms outlined toward
the end of \SRef{sec:decompose}. Throughout the rest of this section we assume that for each buyer $i \in
[\NumAgents]$ we can compute a single buyer mechanism $\Mech_i$ and a corresponding concave benchmark
$\Rev_i$, which together provide $\alpha$-approximation of the optimal single buyer mechanism/optimal
benchmark for buyer $i$. We show that the resulting multi buyer mechanism will be
$\gamma_\NumUnits\alpha$-approximation of the the optimal multi buyer mechanism in $\Mechs$, where $\NumUnits
= \min_j \NumUnits_j$ and $\gamma_\NumUnits$ is the optimal magician parameter which is at least
$1-\frac{1}{\sqrt{\NumUnits+3}}$ (\autoref{def:magician_param}) .

\subsection{Pre-Rounding}
This mechanism serves the buyers sequentially (arbitrary order); for each buyer $i$, it selects a subset
$\ItemSubset_i$ of available items and runs the single buyer mechanism
$\Cons{\Mech_i}{\Induced{\XAAllocC_i}{\ItemSubset_i}}$, where $\XAAllocC_i$ is an optimal assignment for the
benchmark convex program \eqref{P:OPT'}, and  $\Induced{\XAAllocC_i}{\ItemSubset_i}$ denotes the vector
resulting from $\XAAllocC_i$ by zeroing the entries corresponding to items not in $\ItemSubset_i$. In
particular, this mechanism sometimes precludes some items from early buyers to make them available to late
buyers. For each item, the mechanism tries to minimize the probability of preclusion for each buyer by
equalizing it for all buyers. Note that, for any given pair of buyer and item, we only care about the
probability of preclusion in expectation, where the expectation is taken over the types of other buyers and
the random choices of the mechanism. The mechanism is explained in detail in \autoref{def:pre}.

\begin{samepage}
\begin{definition}[$\gamma$-Pre-Rounding] \ \
\label{def:pre}%
\begin{enumerate}[(I)]
\item
Solve the convex program \eqref{P:OPT'} and let $\XAAllocC$ be an optimal assignment.

\item
For each item $j \in \RangeN{\NumItems}$, create an instance of $\gamma$-conservative magician
(\autoref{def:gcm}) with $\NumUnits_j$ wands (this will be referred to as the $j^{th}$ magician). We will use
these magicians through the rest of the mechanism. Note that $\gamma$ is a parameter that is given.

\item
For each buyer $i \in \RangeN{\NumAgents}$:
\begin{enumerate}
\item
For each item $j \in \RangeN{\NumItems}$, write $\XAAllocC_{ij}$ on a box and present it to the $j^{th}$
magician. Let $\ItemSubset_i$ be the set of items where the corresponding magicians opened the box.

\item
Run $\Cons{\Mech_i}{\Induced{\XAAllocC_i}{\ItemSubset_i}}$ on buyer $i$ and use its outcome as the final
outcome for buyer $i$.

\item
For each item $j \in \RangeN{\NumItems}$, if a unit of item $j$ was allocated to buyer $i$ in the previous
step, break the wand of the $j^{th}$ magician.

\end{enumerate}
\end{enumerate}
\end{definition}
\end{samepage}

Note that since $\XAAllocC$ is a feasible assignment for convex program \eqref{P:OPT'}, it must satisfy
$\sum_i \XAAllocC_{ij} \le \NumUnits_j$, so by setting $\gamma \leftarrow \gamma_\NumUnits$ and by
\autoref{thm:magician} and \autoref{def:magician_param} we can argue that each $\ItemSubset_i$ includes each
item $j$ with probability at least $\gamma_\NumUnits$ where $\gamma_\NumUnits$ is at least
$1-\frac{1}{\sqrt{\NumUnits+3}}$.

In order for the above mechanism to retain at least a $\gamma$-fraction of the the expected objective value
of each $\Cons{\Mech_i}{\XAAllocC_i}$, further technical assumptions are needed in addition to $\gamma \le
\gamma_\NumUnits$. We show that it is enough to assume each $\Rev_i$ has a budget-balanced and cross
monotonic cost sharing scheme.

\begin{definition}[Budget Balanced Cross Monotonic Cost Sharing Scheme]
\label{def:cs}%
A function $\Rev : \RangeAB{0}{1}^\NumItems \to \PosReals$ has a budget balanced cross monotonic cost sharing
scheme iff there exists a cost share function $\xi : \RangeN{\NumItems} \times \RangeAB{0}{1}^\NumItems \to
\PosReals$ with the following two properties:
\begin{enumerate}[(i)]
\item
$\xi$ must be budget balanced which means for all $\XAAlloc \in \RangeAB{0}{1}^\NumItems$ and $\ItemSubset
\subseteq \RangeN{\NumItems}$, $\sum_{j \in \ItemSubset} \xi(j, \Induced{\XAAlloc}{\ItemSubset}) =
\Rev(\Induced{\XAAlloc}{\ItemSubset})$.
\item
$\xi$ must be cross monotonic which means for all  $\XAAlloc \in \RangeAB{0}{1}^\NumItems$, $j\in
\RangeN{\NumItems}$ and $\ItemSubset,\ItemSubset' \subseteq \RangeN{\NumItems}$, $\xi(j,
\Induced{\XAAlloc}{\ItemSubset}) \ge \xi(j, \Induced{\XAAlloc}{\ItemSubset \cup \ItemSubset'})$.
\end{enumerate}
\end{definition}

Intuitively, a cost share function associates a fraction of the expected objective value returned by the
benchmark function $\Rev$ to each item; and ensures that the fraction associated with each item does not
decrease when other items are excluded. In particular, the above assumption holds if
$\Rev(\Induced{\XAAlloc}{\ItemSubset})$ is a submodular function of $\ItemSubset$ (e.g., for welfare
maximization, assuming that buyers' valuations are submodular\footnote{We conjecture that it holds in general
for revenue maximization, when buyers' valuations are submodular and $\Mechs$ is restricted to mechanisms
which use buyer specific item pricing.}). Note that it is enough to show that such a cost sharing function
exists; however it is never used in the mechanism and its computation is not required.

\begin{theorem}[$\gamma$-Pre-Rounding]
\label{thm:pre}%
Suppose for each buyer $i$, $\Mech_i$ is an $\alpha$-approximate incentive compatible single buyer mechanism,
and $\Rev_i$ is the corresponding concave benchmark. Also suppose $\Rev_i$ has a budget balanced cross
monotonic cost sharing scheme. Then, for any $\gamma \in \RangeAB{0}{\gamma_\NumUnits}$, the
$\gamma$-pre-rounding mechanism (\autoref{def:pre}) is dominant strategy incentive compatible (DSIC)
mechanism which is in $\Mechs$ and is a $\gamma\alpha$-approximation of the optimal mechanism in $\Mechs$.
\end{theorem}
\begin{proof}
See \SRef{proof:thm:pre}.
\end{proof}

\begin{remark}
The $\gamma$-pre-rounding mechanism assumes no control and no prior information about the order in which buyers are visited. The
order specified in the mechanism is arbitrary and could be replaced by any other ordering which may be unknown in advance. In
particular, this mechanism can be adopted to online settings in which buyers are served in an unknown order.
\end{remark}

\begin{corollary}
In any setting where \autoref{thm:pre} is applicable and when $\Mechs$ includes all feasible BIC mechanisms,
the gap between the optimal DSIC mechanism and the optimal BIC mechanism is at most $1/\gamma_k$. This gap is
at most $2$ (for $k=1$) and vanishes as $k \to \infty$. That is because \autoref{def:pre} is always DSIC, yet
it approximates the optimal mechanism in $\Mechs$.
\end{corollary}

\subsection{Post-Rounding}
\label{sec:pre}

This mechanism runs $\Cons{\Mech_i}{\XAAllocC_i}$ simultaneously and independently for all buyers $i$ to
compute a tentative allocation/payment for each buyer; it then deallocates some of the items at random to
ensure that the supply constraints are met at every instance; it ensures that the probability of deallocation
perceived by each buyer (i.e., in expectation over the types of other buyers and random choices of the
mechanism) is equalized and therefore simultaneously minimized for all buyers. The payments are also adjusted
respectively. The mechanism is explained in detail in \autoref{def:post}.

\begin{samepage}
\begin{definition}[$\gamma$-Post-Rounding] \ \
\label{def:post}%
\begin{enumerate}[(I)]
\item
Solve the convex program \eqref{P:OPT'} and let $\XAAllocC$ denote an optimal assignment.

\item
Run $\Cons{\Mech_i}{\XAAllocC_i}$ simultaneously and independently for all buyers $i \in
\RangeN{\NumAgents}$, and let $\XPAlloc'_i \subseteq \RangeN{\NumItems}$ and $\Pay'_i \in \PosReals$ denote
respectively the allocation (subset of items) and payment computed by $\Cons{\Mech_i}{\XAAllocC_i}$ for buyer
$i$.

\item
For each item $j \in \RangeN{\NumItems}$, create an instance of $\gamma$-conservative magician
(\autoref{def:gcm}) with $\NumUnits_j$ wands (this will be referred to as the $j^{th}$ magician). We will use
these magicians through the rest of the mechanism. Note that $\gamma$ is a parameter that is given.

\item
For each buyer $i \in \RangeN{\NumAgents}$:
\begin{enumerate}
\item
For each item $j \in \RangeN{\NumItems}$, write $\XAAllocA_{ij}$ on a box and present it to the $j^{th}$
magician, where $\XAAllocA_{ij}$ is the exact probability\footnote{Note that $\XAAllocC_{ij}$ is only an
upper bound on the probability of allocation, so $\XAAllocA_{ij} \le \XAAllocC_{ij}$} of
$\Cons{\Mech_i}{\XAAllocC_i}$ allocating a unit of item $j$ to buyer $i$; let $\ItemSubset_i$ be the set of
items where the corresponding magicians opened the box.

\item
Let $\XPAlloc_i \leftarrow \ItemSubset_i \cap \XPAlloc'_i$ and $\Pay_i \leftarrow \gamma\Pay'_i$. The final
allocation and payment of buyer $i$ is given by $\XPAlloc_i$ and $\Pay_i$ respectively.

\item
For each item $j \in \XPAlloc_i$, break the wand of the $j^{th}$ magician.
\end{enumerate}
\end{enumerate}
\end{definition}
\end{samepage}

Note that $\sum_i \XAAllocA_{ij} \le \sum_i \XAAllocC_{ij} \le \NumUnits_j$; so by setting $\gamma \leftarrow
\gamma_\NumUnits$ and by \autoref{thm:magician} and \autoref{def:magician_param} we can argue that each
$\ItemSubset_i$ includes each item $j$ with probability at least $\gamma_\NumUnits$ where $\gamma_\NumUnits$
is at least $1-\frac{1}{\sqrt{\NumUnits+3}}$. Consequently, any item that is in $\XPAlloc'_i$ will also be in
$\XPAlloc_i$ with probability exactly $\gamma$.

In order for $\gamma$-post-rounding to retain at least a $\gamma$-fraction of the the expected objective
value of each $\Cons{\Mech_i}{\XAAllocC_i}$ and preserve incentive compatibility, further technical
assumptions are needed in addition to $\gamma \le \gamma_\NumUnits$; next, we present a set of assumptions
which is sufficient for this purpose\footnote{I.e., one might come up with other sets of assumptions that are
also sufficient.}.

\begin{enumerate}[{($A'$}1{)}]
\item
\label{asm:post:first}%
\label{asm:post:exact_q}%
The exact ex ante allocation rule for each $\Cons{\Mech_i}{\XAAllocC_i}$ (i.e., $\XAAllocA$) must be
available (i.e., efficiently computable). Note that $\XAAllocC$ is only an upper bound on the ex ante
allocation rule.

\item
\label{asm:post:obj}%
The objective functions must be of the form $\Obj_i(\Type_i, \XPAlloc_i,\Pay_i)=\Obj_i(\Type_i, \XPAlloc_i,
0)+\PayFactor_i \Pay_i$ in which $\PayFactor_i \in \PosReals$ is an arbitrary fixed constant. Also, each
$\Obj_i(\Type_i, \XPAlloc_i, 0)$ must have cost sharing scheme in $\XPAlloc_i$ which is cross monotonic and
budget balanced.

\item
\label{asm:post:bic}%
The resulting mechanism must be in $\Mechs$. In particular, that implies $\Mechs$ may not be restricted to
any from of incentive compatibility stronger than Bayesian incentive compatibility (BIC), because the
$\gamma$-post-rounding is only BIC.

\item
\label{asm:post:valuations}%
The valuations of each buyer must be in the form of a weighted rank function of some matroid.
\label{asm:post:last}%
\end{enumerate}

Observe that \myref{A'}{asm:post:obj} obviously holds for revenue maximization (because $\Obj_i(\Type_i,
\XPAlloc_i, \Pay_i)=\Pay_i$), and also for welfare maximization with quasilinear utilities and submodular
valuations (because $\Obj_i(\Type_i, \XPAlloc_i, \Pay_i)=\Valuation_i(\Type_i, \XPAlloc_i)$ where
$\Valuation_i(\Type_i, \XPAlloc_i)$ is the valuation of buyer $i$ for allocation $\XPAlloc_i$\footnote{Note
that the payment terms cancel out because the utility of the seller is counted toward the social welfare of
the mechanism}). Next, we formally define \myref{A'}{asm:post:valuations}.

\begin{definition}[Matroid Weighted Rank Valuation]
\label{def:matval}%
A valuation function $\Valuation : 2^\NumItems \to \PosReals$ is a matroid weighted rank valuation if{f}
there exists a matroid $\Mat=(\RangeN{\NumItems}, \IndepSets)$, and a weight function $\Weight :
\RangeN{\NumItems} \to \PosReals$ such that $\Valuation(\ItemSubset)$ is equal to the weight of a maximum
weight independent subset of $\ItemSubset$, i.e,

\begin{align*}
    \Valuation(\ItemSubset) &= \max_{\substack{\IndepSet \in \IndepSets \cap 2^\ItemSubset}}  \sum_{j \in
    \IndepSet} \Weight(j), & & \forall \ItemSubset \subseteq \RangeN{\NumItems}
\end{align*}
\end{definition}

Matroid weighted rank valuations include additive valuations with demand constraints, unit demand valuations,
etc.

\begin{theorem}[$\gamma$-Post-Rounding]
\label{thm:post}%
Suppose for each buyer $i$, $\Mech_i$ is an $\alpha$-approximate incentive compatible single buyer mechanism,
and $\Rev_i$ is a corresponding concave benchmark. Also suppose the assumptions \myref{A'}{asm:post:first}
through \myref{A'}{asm:post:last} hold. Then, for any $\gamma \in \RangeAB{0}{\gamma_\NumUnits}$, the
$\gamma$-post-rounding mechanism (\autoref{def:post}) is a Bayesian incentive compatible (BIC) mechanism
which is in $\Mechs$ and is a $\gamma\alpha$-approximation of the optimal mechanism in $\Mechs$.
\end{theorem}
\begin{proof}
See \SRef{proof:thm:post}.
\end{proof}

\section{Single Buyer Mechanisms}
\label{sec:single}%

In this section, we present approximately optimal single buyer mechanisms for several common settings. Each
one of the single buyer mechanisms presented in this section satisfies the requirements of one of the generic
multi buyer mechanisms of \SRef{sec:expansion}, so they can be readily converted to a multi buyer mechanisms.
Except for \SRef{sec:corr}, we restrict the space of mechanisms to item pricing mechanisms with budget
randomization as defined next.

\begin{definition}[Item Pricing with Budget Randomization (\IP)]
\label{def:aip}%
An item pricing mechanism is a possibly randomized mechanism that offers a menu of prices to each buyer and
allows each buyer to choose their favorite bundle. The payment of a buyer is equal to the total price of the
items in her purchased bundle. Note that the prices offered to different buyers do not need to be identical
and buyers can be served sequentially. In the presence of budget constraints, a buyer is allowed to pay a
fraction of the price of an item and receive the item with a probability equal to the paid
fraction\footnote{A utility maximizing buyer, with submodular valuations and budget constraint, always pays
the full price for any item she purchases, except potentially for the last item purchased, for which she must
have run out of budget.}. A mechanism is considered an item pricing mechanism if its outcome can be
interpreted as such\footnote{I.e., an item pricing mechanism may collect all the reports and compute the
final outcome along with buyer specific prices, such that the outcome of each buyer would be the same as if
each buyer purchased their favorite bundle according to her observed prices, and the prices observed by each
buyer should be independent of her report.}.
\end{definition}

Item pricing mechanisms are simple and practical as opposed to optimal BIC mechanisms which often involve
lotteries. Also budget randomization allows us to get around the hardness of the knapsack problem faced by
the budgeted buyers; in particular, assuming that budgets are large compared to prices, budget randomization
can be safely ignored since the optimal integral solution of the knapsack problem approaches its optimal
fractional solution.

Table \ref{tab:results} lists several settings for which we obtain a multi buyer mechanism with an improved
approximation factor compared to previous best known approximations. For each setting, we present a single
buyer mechanism that satisfies the requirements of one of the generic multi buyer mechanisms of
\SRef{sec:expansion}. The corresponding single buyer mechanisms are presented in detail throughout the rest
of this section. Note that the final approximation factor for each multi buyer mechanism is equal to the
approximation factor of the corresponding single buyer mechanism multiplied by $\gamma_\NumUnits$; recall
that $\gamma_\NumUnits \ge 1-\frac{1}{\sqrt{\NumUnits + 3}}$ which approaches $1$ as $\NumUnits \to \infty$.

\begin{table}[h]
\footnotesize
\begin{tabular}{|p{0.4\columnwidth}|c|c|p{0.05\columnwidth}|}
\hline Setting &  Approx & Space of Mechanisms & Ref \\
\hline single item(multi unit), unit demand, budget constraint, revenue maximization & $\gamma_\NumUnits$ & item pricing with budget randomization & \SRef{sec:budget1} \\
\hline multi item(heterogenous), unit demand, product distribution, revenue maximization & $\frac{1}{2}\gamma_\NumUnits$ & deterministic & \SRef{sec:unit} \\
\hline multi item(heterogenous), additive valuations, product distribution, budget constraint, revenue maximization & $(1-\frac{1}{e})\gamma_\NumUnits$ & item pricing with budget randomization & \SRef{sec:budgetm} \\
\hline multi item(heterogenous), additive valuations, correlated distribution with polynomial number of types, budget constraint, matroid constrains, revenue or welfare maximization & $\gamma_\NumUnits$ & randomized (BIC) & \SRef{sec:corr}  \\
\hline
\end{tabular}
\caption{Summary of mechanisms obtained using the framework of this paper. \label{tab:results}}
\end{table}

For each single buyer mechanism presented in this paper, the single buyer benchmark function $\Rev(\XAAlloc)$
is defined as the optimal value of some convex program of the following general form, in which $\SObj$ is
some concave function, $g_j(\wdot)$ are some convex functions, and $\mathds{Y}$ is some convex polytope (in
the rest of this section we only consider a single buyer, so we will omit the subscript $i$).

\begin{alignat*}{3}
    \text{maximize} &   &\qquad & \SObj(y) \tag{$\overline{OPT}_1$} & \qquad &\label{CP:u} \\
    \text{subject to}&  &       & g_j(y) \le \XAAllocC_j, & & \forall j \in \RangeN{\NumItems}  \\
                    &   &       & y \in \mathds{Y}
\end{alignat*}

\begin{lemma}
\label{lem:concave}%
$\Rev(\XAAllocC)$ is concave, i.e., the optimal value of a convex program of the form \eqref{CP:u} is always
concave in $\XAAllocC$.
\end{lemma}
\begin{proof}
See \autoref{proof:lem:concave}.
\end{proof}

Note that we can substitute each $\Rev_i(\wdot)$ in the multi buyer benchmark convex program \eqref{P:OPT'}
with the corresponding single buyer benchmark convex program to obtain a combined convex program which can be
solved efficiently. If each $\Rev_i$ is captured by a linear program, the combined multi buyer program will
also be a linear program.

\subsection{Single Item, Unit Demand, Budget Constraint}
\label{sec:budget1}%

In this section, we consider a unit-demand buyer with a publicly known budget $\Budget$ and one type of item
(i.e., $\NumItems=1$). The only private information of the buyer is her valuation for the item, which is
drawn from a publicly known distribution with CDF $\CDF(\wdot)$. To avoid complicating the proofs, we assume
that $\CDF(\wdot)$ is continuous and strictly increasing in its domain\footnote{The proofs can be modified to
work without this assumption.}. We present a single buyer mechanism which is optimal among item pricing
mechanisms with budget randomization (\IP). We start by defining the modified CDF function $\CDFB(\wdot)$ as
follows.
\begin{align*}
    \CDFB(\Val) &=
        \begin{cases}
        \CDF(\Val)                    & \Val \le \Budget \\
        1-(1-\CDF(\Val))\frac{\Budget}{\Val}   & \Val \ge \Budget
        \end{cases} \tag{$\CDFB$} \label{eq:CDFB}
\end{align*}

Intuitively, $1-\CDFB(\Price)$ is the probability of allocating the item to the buyer if we offer the item at
price $\Price$. Note that the buyer only buys if her valuation is more than $\Price$ which happens with
probability $1-\CDF(\Price)$ ; if $\Price > \Budget$, she will pay her whole budget and only get the item
with probability $\frac{\Budget}{\Price}$, otherwise she pays the full price and receives the item with
probability 1. Observe that if we want to allocate the item with probability $\XAAlloc$ we can offer a price
of ${\CDFB}^{(-1)}(1-\XAAlloc)$ which yields a revenue of $\XAAlloc  {\CDFB}^{(-1)}(1-\XAAlloc)$ in
expectation. Define $\SRev(\XAAlloc)=\XAAlloc {\CDFB}^{(-1)}(1-\XAAlloc)$ and let $\SRevCC(\XAAlloc)$ denote
its concave closure (i.e., the smallest concave function that is an upper bound on $\SRev(\XAAlloc)$ for
every $\XAAlloc$). We will address the problem of efficiently computing $\SRevCC(\XAAlloc)$ later in
\autoref{lem:cc}. Next, we show that the optimal value of the following convex program is equal to the
expected revenue of the optimal single buyer \IP{} mechanism subject to $\XAAllocC$; therefore we will define
the single buyer benchmark function $\Rev(\XAAllocC)$ to be equal to the optimal value of this program as a
function of $\XAAllocC$.

\begin{alignat*}{3}
    \text{maximize} &   &\qquad & \SRevCC(\XAAlloc) \tag{$\SRevSC_{single}$} \label{CP:rev1} \\
    \text{subject to}&  &       & \XAAlloc \le \XAAllocC  \\
                    &   &       & \XAAlloc \ge 0
\end{alignat*}

\begin{theorem}
\label{thm:rev1}%
The revenue of the optimal single buyer \IP{} mechanism, subject to an upper bound of $\XAAllocC$ on the ex
ante allocation rule, is equal to the optimal value of the convex program \eqref{CP:rev1}. Furthermore,
assuming that $\XAAlloc^*$ is the optimal assignment for the convex program, if
$\SRevCC(\XAAlloc^*)=\SRev(\XAAlloc^*)$, then the optimal mechanism uses a single price
$\Price={\CDFB}^{(-1)}(1-\XAAlloc^*)$ otherwise, it randomized between two prices $\Price^-, \Price^+$ with
probabilities $\theta$ and $1-\theta$ for some $\theta \in \RangeAB{0}{1}$ and $\Price^-, \Price^+$.
\end{theorem}
\begin{proof}
First, we prove that the expected revenue of the optimal single buyer \IP{} mechanism, subject to
$\XAAllocC$, is upper bounded by $\SRevCC(\XAAlloc^*)$. We then construct a price distribution that obtains
this revenue. Note that any single buyer \IP{} mechanism can be specified as a distribution over prices. Let
$\PriceDist$ be the optimal price distribution. So the optimal revenue is $\Ex[\Price \sim \PriceDist]{\Price
(1-\CDFB(\Price))}$. Note that every price $\Price$ corresponds to an allocation probability
$\XAAlloc=1-\CDFB(\Price)$. So any probability distribution over $\Price$ can be specified as a probability
distribution over $\XAAlloc$. Let $\AllocDist$ denote the probability distribution over $\XAAlloc$ that
corresponds to price distribution $\PriceDist$, so we can write
\begin{align*}
    \text{optimal revenue} &= \Ex[\XAAlloc \sim \AllocDist]{\XAAlloc  {\CDFB}^{(-1)}(1-\XAAlloc)}
        = \Ex[\XAAlloc \sim \AllocDist]{\SRev(\XAAlloc)} \\
        &\le \Ex[\XAAlloc\sim \AllocDist]{\SRevCC(\XAAlloc)}
        \le \SRevCC(\Ex[\XAAlloc \sim \AllocDist]{\XAAlloc}) & &\text{By Jensen's inequality}
\end{align*}
which means the optimal revenue is upper bounded by the value of the convex program for
$\XAAlloc=\Ex[\XAAlloc \sim \AllocDist]{\XAAlloc}$ \footnote{Note that $\Ex[\XAAlloc \sim
\AllocDist]{\XAAlloc}$ is exactly the probability of allocating the item by the price distribution
$\PriceDist$, so it must be no more than $\XAAllocC$}; so the optimal revenue is upper bounded by the optimal
value of the convex program. That completes the first part of the proof.

Next, we construct an optimal price distribution. If $\SRevCC(\XAAlloc^*)=\SRev(\XAAlloc^*)$, the optimal
price distribution is just a single price $\Price = {\CDFB}^{(-1)}(1-\XAAlloc^*)$; otherwise, by definition
of concave closure, there are two points $\XAAlloc^-$ and $\XAAlloc^+$ and $\theta \in \RangeAB{0}{1}$ such
that $\XAAlloc^*=\theta \XAAlloc^- + (1-\theta)\XAAlloc^+$ and $\SRevCC(\XAAlloc^*)=\theta
\SRev(\XAAlloc^-)+(1-\theta)\SRev(\XAAlloc^+)$. In the latter case, the optimal price distribution offers
price $\Price^- = {\CDFB}^{(-1)}(1-\XAAlloc^-)$ with probability $\theta$ and offers price $\Price^+ =
{\CDFB}^{(-1)}(1-\XAAlloc^+)$ with probability $1-\theta$.
\end{proof}

Formally, an optimal single buyer \IP{} mechanism can be constructed as follows.

\begin{samepage}
\begin{definition}[Mechanism] \ \
\label{def:budget1}%
\begin{itemize}
\item
Define the single buyer benchmark $\Rev(\XAAllocC)$ to be the optimal value of the convex program
\eqref{CP:rev1} as a function of $\XAAllocC$.

\item
Given $\XAAllocC$, solve \eqref{CP:rev1} and let $\XAAlloc$ be an optimal assignment.

\item
If $\SRevCC(\XAAlloc)=\SRev(\XAAlloc)$, offer the single price $\Price={\CDFB}^{(-1)}(1-\XAAlloc)$, otherwise
randomize between two prices $\Price^-$ and $\Price^+$ as explained in the proof of \autoref{thm:rev1}.
\end{itemize}
\end{definition}
\end{samepage}

\begin{theorem}
The mechanism of \autoref{def:budget1} is the optimal revenue maximizing single buyer \IP{} mechanism.
Furthermore, this mechanism satisfies the requirements of $\gamma$-pre-rounding.
\end{theorem}
\begin{proof}
The proof of the optimality follows from \autoref{thm:rev1}. Furthermore, the benchmark function,
$\Rev(\XAAlloc)$, is concave (this follows from \autoref{lem:concave}) and it has a trivial budget balanced
cost sharing scheme (because there is only one item), therefore it meets the requirements of
$\gamma$-pre-rounding.
\end{proof}

Next, we address the problem of efficiently computing $\SRevCC(\wdot)$.

\begin{lemma}
\label{lem:cc}%
A $(1+\epsilon)$-approximation of $\SRevCC(\wdot)$, which we denote by $\SRevCC_{1+\epsilon}(\wdot)$, can be
constructed using a piece-wise linear function with $\ell = \frac{\log L}{\log (1+\epsilon)}$ pieces and in
time $O(\ell \log \ell)$ in which $L$ is the ratio of the maximum valuation to minimum non-zero valuation.
Note that we need at least $\log_2 L$ bits just to represent such valuations so this construction is
polynomial in the input size for any constant $\epsilon$.
\end{lemma}
\begin{proof}
WLOG, assume that all possible non-zero valuations of the buyer are in the range of $\RangeAB{1}{L}$. Let
$\ell = \lfloor\frac{\log L}{\log(1+\epsilon)}\rfloor$. For $r= 0 \cdots \ell$, consider the prices
$\Price_r=(1+\epsilon)^{\ell-r}$ and compute the corresponding $\XAAlloc_r=1-\CDFB(\Price_r)$. Construct
$\SRevCC_{1+\epsilon}(\wdot)$ by constructing the convex hall of the points:
\\
$(0,0), (\XAAlloc_1, \Price_1\XAAlloc_1), (\XAAlloc_2, \Price_2\XAAlloc_2),\ldots, (\XAAlloc_\ell,
\Price_\ell \XAAlloc_\ell), (1, 0)$. This can be done in time $O(\ell \log \ell)$. Note that
${\CDFB}^{(-1)}(1-\XAAlloc)$ is a decreasing function of $\XAAlloc$ so at every $\XAAlloc\in
\RangeAB{\XAAlloc_r}{\XAAlloc_{r+1}}$, the corresponding price is ${\CDFB}^{(-1)}(\XAAlloc) \in
\RangeAB{\Price_{r+1}}{\Price_{r}}$ but $\Price_{r}=(1+\epsilon)\Price_{r+1}$ therefore at every $\XAAlloc$,
$\SRev_{1+\epsilon}(\XAAlloc) \le \SRevCC(\XAAlloc) \le (1+\epsilon)\SRev_{1+\epsilon}(\XAAlloc)$ which
completes the proof.
\end{proof}

\begin{remark}
In order to use $\SRev_{1+\epsilon}(\wdot)$ in the single buyer mechanism of \autoref{def:budget1}, we need
to substitute $(1+\epsilon)\SRevCC_{1+\epsilon}(\wdot)$ in the objective function of the convex program
\eqref{CP:rev1} instead of $\SRevCC(\wdot)$ for computing the benchmark. Furthermore, the mechanism will be a
$(1-\epsilon)$-approximation of the optimal single buyer \IP{} mechanism. Also notice that finding $\Price^-$
and $\Price^+$ from $\SRev_{1+\epsilon}(\wdot)$ is trivial.
\end{remark}

\subsection{Multi Item (Independent), Unit Demand}
\label{sec:unit}%

In this section, we consider a unit demand buyer with private independent valuations for $\NumItems$ items.
We assume that for each item $j$, the buyer's valuation is distributed independently according to a publicly
known distribution with CDF $\CDF_j(\wdot)$. We present a single buyer mechanism which is a
$\frac{1}{2}$-approximation of the optimal deterministic revenue maximizing mechanism. To avoid complicating
the proofs, we assume that each $\CDF_j(\wdot)$ is continuous and strictly increasing in its domain.
Furthermore, we require the distributions to be regular. This mechanism can be used with
$\gamma$-pre-rounding (\autoref{def:pre}) to yield a $\frac{1}{2}\gamma_\NumUnits$-approximate sequential
posted pricing multi buyer mechanism. The previous best approximation mechanism for this setting was a
$\frac{1}{6.75}$-approximate sequential posted pricing mechanism by \cite{CHMS10}\footnote{Note that the
mechanism of \cite{CHMS10} does not work for non-regular distributions despite the authors' claim.}.


We start by defining $\SRev_j(\XAAlloc) = \XAAlloc  \CDF_j^{-1}(1-\XAAlloc)$ for each item $j$. Because
$\CDF_j(\wdot)$ is corresponds to a regular distribution, $\SRev_j(\wdot)$ is concave as shown in the
following lemma.

\begin{lemma}
If $\CDF(\wdot)$ is the CDF of a regular distribution, the function $\SRev(\XAAlloc)=\XAAlloc
\CDF^{-1}(1-\XAAlloc)$ is concave.
\end{lemma}
\begin{proof}
It is enough to show that $\frac{\partial}{\partial \XAAlloc} \SRev(\XAAlloc)$ is non-increasing in
$\XAAlloc$. Observe that $\frac{\partial}{\partial \XAAlloc} \SRev(\XAAlloc) = \CDF^{-1}(1-\XAAlloc) -
\frac{\XAAlloc}{\PDF(\CDF^{-1}(1-\XAAlloc))}$ in which $\PDF(\wdot)$ is the derivative of $\CDF(\wdot)$. By
substituting $\XAAlloc=1-\CDF(\Price)$, it is enough to show that the resulting function is non-decreasing in
$\Price$ because $\XAAlloc$ is itself non-increasing in $\Price$. However, by this substitution we get
$\frac{\partial}{\partial \XAAlloc} \SRev(\XAAlloc)=\Price-\frac{1-\CDF(\Price)}{\PDF(\Price)}$ which is
non-decreasing in $\Price$ by definition of regularity.
\end{proof}

Note that any deterministic mechanism for a unit demand buyer can be interpreted as item pricing.
Consequently, $\SRev_j(\XAAlloc_j)$ is the maximum revenue that such a mechanism can obtain if item $j$ is
allocated with probability $\XAAlloc_j$. Next, we show that the following convex program gives an upper bound
the on the expected optimal revenue.

\begin{alignat*}{3}
    \text{maximize} &   &\qquad \sum_j & \SRev_j(\XAAlloc_j) & \qquad & \tag{$\SRevSC_{unit}$} \label{CP:revu} \\
    \text{subject to}&  &       \XAAlloc_j & \le \XAAllocC_j, & & \forall j \in \RangeN{\NumItems} \tag{$\lambda_j$} \label{eq:lambda_j} \\
                    &   &       \sum_j \XAAlloc_j & \le 1 \tag{$\tau$} \label{eq:tau} \\
                    &   &       \XAAlloc_j &\ge 0, & & \forall j \in \RangeN{\NumItems}  \tag{$\mu_j$} \label{eq:mu_j}
\end{alignat*}

\begin{theorem}
\label{thm:revu}%
The revenue of the optimal deterministic single buyer mechanism, subject to an upper bound of $\XAAllocC$ on
the ex ante allocation rule, is no more than the optimal value of the convex program \eqref{CP:revu}.
\end{theorem}
\begin{proof}
Let $\XAAlloc^*$ be the ex ante allocation rule of the optimal single buyer deterministic mechanism. So the
expected revenue obtained from each item $j$ is upper bounded by $\SRev_j(\XAAlloc^*_j)$ (proof of this claim
is essentially the same as the proof of \autoref{thm:rev1}). Consequently, the expected optimal revenue
cannot be more that $\sum_j \SRev_j(\XAAlloc^*_j)$. Furthermore, the optimal mechanism never allocates more
than one item, so $\sum_j \XAAlloc^*_j \le 1$, and also $\XAAlloc^*_j \le \XAAllocC_j$; therefore
$\XAAlloc^*$ is a feasible solution for the convex program; so the expected optimal revenue is upper bounded
by the optimal value of the convex program.
\end{proof}

Next, we present the single buyer mechanism.

\begin{samepage}
\begin{definition}[Mechanism] \ \
\label{def:unit}%
\begin{itemize}
\item

Define the benchmark $\Rev(\XAAllocC)$ to be the optimal value of \eqref{CP:revu} as a function of
$\XAAllocC$.

\item
Given $\XAAllocC$, solve \eqref{CP:revu} and let $\XAAlloc$ denote an optimal assignment.

\item
For each item $j$, assign the price $\Price_j = \CDF_j^{-1}(1-\XAAlloc_j)$. WLOG, assume that items are
indexed in non-decreasing order of prices, i.e., $\Price_1 \le \ldots \le \Price_\NumItems$.

\item
For each item $j$, define $\TailRev_j = \max(\XAAlloc_j \Price_j + (1-\XAAlloc_j)\TailRev_{j+1},
\TailRev_{j+1})$ and let $\TailRev_{\NumItems+1} = 0$. Let $\ItemSubset^*$ be the subset of items defined as
$\ItemSubset^* = \{j | \Price_j \ge \TailRev_{j+1} \}$.

\item Only offer the items in $\ItemSubset^*$ at prices computed in the previous step (i.e., set the price of other items
to infinity).
\end{itemize}
\end{definition}
\end{samepage}

\begin{theorem}
The mechanism of \autoref{def:unit} obtains at least $\frac{1}{2}$ of the revenue of the optimal
deterministic single buyer mechanism in expectation. Furthermore, it satisfies the requirements of
$\gamma$-pre-rounding.
\end{theorem}
\begin{proof}
First, we show that this mechanism obtains in expectation at least $\frac{1}{2}$ of its benchmark
$\Rev(\XAAllocC)$, which by \autoref{thm:revu} is an upper bound on the optimal revenue. Observe that
$\Rev(\XAAllocC) = \sum_j \XAAlloc_j \Price_j$ where $\XAAlloc_j$ is exactly the probability that the
valuation of the buyer for item $j$ is at least $\Price_j$. Now consider an ``adversary replica'' who has the
exact same valuations as the original buyer, but always buys the item that has the lowest price among all the
items priced below her valuation. For any assignment of prices, the revenue obtained from the adversary
replica is a lower bound on the revenue obtained from the original buyer. So it is enough to show that the
mechanism obtains a revenue of at least $\frac{1}{2}\sum_j \XAAlloc_j \Price_j$ from the adversary replica.
Observe that $\TailRev_j$ is exactly the expected revenue obtained from the adversary replica when offered
the items in $\ItemSubset^* \cap \RangeMN{j}{\NumItems}$. In particular, item $j$ is included in
$\ItemSubset^*$ if $\Price_j \ge \Rev_{j+1}$, which implies that the revenue obtained from the purchase of
item $j$, conditioned on purchase, is more than the lower bound on the expected revenue obtained from items
$\RangeMN{j}{\NumItems}$. Finally, observe that the expected revenue obtained from the adversary replica is
exactly $\TailRev_1$. By \autoref{lem:dp} we can conclude that $\TailRev_1 \ge \frac{1}{2}\sum_j \XAAlloc_j
\Price_j $ which completes the proof of the first claim.

Next, we show that this mechanism satisfies the requirements of $\gamma$-pre-rounding. Observe that by
\autoref{lem:concave}, the optimal value of \eqref{CP:revu} is a concave function of $\XAAllocC$; so
$\Rev(\XAAllocC)$ is concave. It only remains to show that $\Rev(\wdot)$ has a budget balanced cross
monotonic cost sharing scheme. Let $\XAAlloc_j(\XAAllocC)$ denote the optimal assignment of variable
$\XAAlloc_j$, in the convex program \eqref{CP:revu}, as a function of $\XAAllocC$. Define the cost share
function
\begin{align*}
    \xi(j, \XAAllocC)=\SRev_j(\XAAlloc_j(\XAAllocC)).
\end{align*}
We shall show that $\xi$ is budget balanced and cross monotonic (see \autoref{def:cs}).
\begin{itemize}
\item
\About{Budget balance} We shall show that for any $\XAAllocC \in \RangeAB{0}{1}^\NumItems$ and any
$\ItemSubset \subseteq \RangeN{\NumItems}$, $\Rev(\Induced{\XAAllocC}{\ItemSubset}) = \sum_{j \in
\ItemSubset} \xi(j,\Induced{\XAAllocC}{\ItemSubset})$. Note that
$\Rev(\Induced{\XAAllocC}{\ItemSubset})=\sum_j \SRev_j(\XAAlloc_j(\Induced{\XAAllocC}{\ItemSubset}))=\sum_{j
\in \ItemSubset} \xi(j, \XAAlloc_j(\Induced{\XAAllocC}{\ItemSubset}))$ which proved that $\xi$ is budget
balanced. Note that $\SRev_j(\XAAlloc_j(\Induced{\XAAllocC}{\ItemSubset}))=0$, for any $j \not \in
\ItemSubset$, because $\XAAlloc_j(\Induced{\XAAllocC}{\ItemSubset})$ is forced to be $0$.

\item
\About{Cross monotonicity} We shall show that $\xi(j,\Induced{\XAAllocC}{\ItemSubset}) \ge
\xi(j,\Induced{\XAAllocC}{\ItemSubset \cup \ItemSubset'})$, for any $\XAAllocC \in \RangeAB{0}{1}^\NumItems$
and any $\ItemSubset, \ItemSubset' \subseteq \RangeN{\NumItems}$. Let the Lagrangian of \eqref{CP:revu} be
defined as follows.
\begin{align*}
    \Lag(\XAAlloc,\lambda,\tau,\mu) &= -\sum_j \SRev_j(\XAAlloc_j)+\sum_j \lambda_j\sdot(\XAAlloc_j-\XAAllocC_j)+ \tau\sdot(\sum_j \XAAlloc_j-1)-\sum_j \mu_j \XAAlloc_j
\end{align*}
The high level idea of the proof is as follows. We show that there is more pressure on the constraint
associated with $\tau$ when the set of available items is $\ItemSubset \cup \ItemSubset'$ instead of
$\ItemSubset$ (i.e., $\tau$ is larger for $\ItemSubset \cup \ItemSubset'$); we then show that the optimal
$\XAAlloc_j$ can be determined from $\tau$; in particular, we show that, as the optimal $\tau$ increases, the
optimal $\XAAlloc_j$ decreases, and consequently $\xi(j, \XAAlloc)$ (which is equal to $\SRev_j(\XAAlloc_j)$)
decreases as well, which proves $\xi$ is cross monotonic. Next we present the proof in detail.

By KKT stationarity conditions, at the optimal assignment the following holds.
\begin{align*}
    \frac{\partial}{\partial \XAAlloc_j} \Lag(\XAAlloc,\lambda,\tau,\mu)& =- \frac{\partial}{\partial \XAAlloc_j} \SRev_j(\XAAlloc_j)+\lambda_j + \tau-\mu_j = 0
\end{align*}

First we show that the optimal $\XAAlloc_j$, and consequently $\xi(j, \XAAllocC)$, can be determined from the
optimal $\tau$; and they are both non-increasing in $\tau$. Observe that (a) all dual variables must be
non-negative, (b) by complementary slackness $\lambda_j$ may be non-zero only if $\XAAlloc_j=\XAAllocC_j$,
and (c) complementary slackness implies that $\mu_j$ may be non-zero only if $\XAAlloc_j=0$; therefore, if
the optimal $\tau$ is given, the optimal assignment for $\XAAlloc_j$ is uniquely\footnote{To avoid
complicating the proof, we assume that the functions $\SRev_j(\wdot)$ are strictly concave, however this
assumption is not necessary.} determined by the above equation and the aforementioned complementarity
slackness conditions. Let $\XAAlloc_j(\tau)$ denote the optimal assignment of $\XAAlloc_j$ as a function of
$\tau$. Due to the concavity of $\SRev_j(\wdot)$, and the above KKT condition, we can argue that
$\XAAlloc_j(\tau)$ is non-increasing in $\tau$, which also implies that $\xi(j, \XAAllocC)$ is non-increasing
in $\tau$.

Next, we prove by contradiction that $\xi$ is cross monotonic. Let $\tau(\XAAllocC)$ denote the optimal
assignment of $\tau$ as a function of $\XAAllocC$. By contradiction, suppose $\xi$ is not cross monotonic,
i.e. $\xi(j^*, \Induced{\XAAllocC}{\ItemSubset \cup \ItemSubset'}) > \xi(j^*,
\Induced{\XAAllocC}{\ItemSubset})$ for some item $j^*$; therefore $\tau(\Induced{\XAAllocC}{\ItemSubset}) >
\tau(\Induced{\XAAllocC}{\ItemSubset \cup \ItemSubset'}) \ge 0$. Since
$\tau(\Induced{\XAAllocC}{\ItemSubset}) > 0$, the inequality associated with $\tau$ must be tight (by
complementary slackness), so $\sum_j \XAAlloc_{j}(\tau(\Induced{\XAAllocC}{\ItemSubset}))=1$. On the other
hand, for all $j$, $\XAAlloc_{j}(\tau(\Induced{\XAAllocC}{\ItemSubset \cup \ItemSubset'})) \ge
\XAAlloc_{j}(\tau(\Induced{\XAAllocC}{\ItemSubset}))$, with the inequality being strict for $j=j^*$, which
means $\sum_j \XAAlloc_{j}(\tau(\Induced{\XAAllocC}{\ItemSubset \cup \ItemSubset'})) > 1$, which is a
contradiction.
\end{itemize}
\end{proof}

\begin{lemma}
\label{lem:dp}%
Let $\Price_1, \ldots, \Price_\NumItems$ and $\XAAlloc_1,\ldots, \XAAlloc_\NumItems$ be two sequences of
non-negative real numbers and suppose $\sum_j \XAAlloc_j \le 1$. For each $j \in \RangeN{\NumItems}$, define
$\TailRev_j = \max(\XAAlloc_j \Price_j + (1-\XAAlloc_j)\TailRev_{j+1}, \TailRev_{j+1})$ and let
$\TailRev_{\NumItems+1}=0$. Then $\TailRev_1 \ge \frac{1}{2}\sum_j \XAAlloc_j \Price_j$.
\end{lemma}
\begin{proof}
See \autoref{proof:lem:dp}.
\end{proof}

\subsection{Multi Item (Independent), Additive, Budget Constraint}
\label{sec:budgetm}%

In this section, we consider a buyer with publicly known budget $\Budget$ who has private independent and
additive valuations for $\NumItems$ items (i.e., her valuation for a bundle of items is the sum of her
valuations for individual items in the bundle). We assume the buyer's valuation for each item $j$ is
distributed independently according to a publicly known distribution with CDF $\CDF_j(\wdot)$. To avoid
complicating the proofs, we assume that each $\CDF_j(\wdot)$ is continuous and strictly increasing in its
domain\footnote{The proofs can be modified to work without this assumption.}. We present a single buyer
mechanism which is a $(1-\frac{1}{e})$-approximation of the optimal revenue maximizing item pricing mechanism
with budget randomization (\IP). This mechanism can be used with $\gamma$-pre-rounding (\autoref{def:pre}) to
yield a $(1-\frac{1}{e})\gamma_\NumUnits$-approximate sequential posted pricing multi buyer mechanism. The
previous best approximation mechanism for this setting was an $O(1)$-approximate\footnote{$\frac{1}{96}$}
sequential posted pricing mechanism by \cite{BGGM10}. We should note that the mechanism in \cite{BGGM10} is
more general as it allows the buyers to have demand constraints as well, and it does not allow for budget
randomization.

As in \SRef{sec:budget1}, we start by defining the modified CDF function $\CDFB_j(\wdot)$ for each item $j$
as follows.
\begin{align*}
    \CDFB_j(\Val) &=
        \begin{cases}
        \CDF_j(\Val)                    & \Val \le \Budget \\
        1-(1-\CDF_j(\Val))\frac{\Budget}{\Val}   & \Val \ge \Budget
        \end{cases} \tag{$\CDFB_j$} \label{eq:CDFBj}
\end{align*}

Furthermore, for each item $j$, let $\SRev_j(\XAAlloc)=\XAAlloc  {\CDFB_j}^{-1}(1-\XAAlloc)$ and let
$\SRevCC_j(\wdot)$ be its concave closure as define in \SRef{sec:budget1}. Also, for each $j$, define
$\Rev_j(\XAAllocC_j)$ to be the optimal value of the following convex program as a function of $\XAAllocC_j$.

\begin{alignat*}{3}
    \text{maximize} &   &\qquad & \SRev_j(\XAAlloc_j) & \qquad & \tag{$\SRevSC_{add}$} \label{CP:revj} \\
    \text{subject to}&  &       \XAAlloc_j & \le \XAAllocC_j \\
                    &   &       \XAAlloc_j &\ge 0
\end{alignat*}

The next theorem provides an upper bound on the revenue of the optimal single buyer \IP{} mechanism.

\begin{theorem}
\label{thm:revm}%
The revenue of the optimal single buyer item pricing mechanism with budget randomization (\IP), subject to an
upper bound of $\XAAllocC$ on the ex ante allocation rule, is no more than $\min(\sum_j \Rev_j(\XAAllocC_j),
\Budget)$, .
\end{theorem}
\begin{proof}
For any $j$, if we were only to sell the item $j$, by \autoref{thm:rev1}, the maximum revenue we could obtain
using an \IP{} mechanism would be no more than $\Rev_j(\XAAllocC_j)$. Observe that if we compute the optimal
price distribution for each item separately, we might only get less revenue because the budget is shared
among all items and the buyer might not be able to buy some of the items that she would otherwise buy if
there were no other items. That means the actually probability of allocating each item $j$ could be less than
the optimal assignment of $\XAAlloc_j$ for the convex program \eqref{CP:revj}; so the optimal joint price
distribution might sell at lower prices; but the extra revenue may only come from lower types which were
originally excluded by the optimal single item mechanism. Consequently, the overall revenue from each item
$j$ cannot be more than $\Rev_j(\XAAllocC_j)$. Finally, observe that the expected revenue of the mechanism
cannot be more that $\Budget$, so it can be no more than $\min(\sum_j \Rev_j(\XAAlloc_j), \Budget)$.
\end{proof}

Next, we present $(1-\frac{1}{e})$-approximate revenue maximizing single buyer \IP{} mechanism.

\begin{samepage}
\begin{definition}[Mechanism] \ \
\label{def:budgetm}%
\begin{itemize}
\item
Define the benchmark $\Rev(\XAAllocC)=\min(\sum_j \Rev_j(\XAAllocC_j), \Budget)$.

\item
Given $\XAAllocC$, solve the convex program of \eqref{CP:revj} for each item $j$, and let $\XAAlloc_j$ denote
an optimal assignment.

\item
For each item $j$, if $\SRevCC_j(\XAAlloc_j)=\SRev_j(\XAAlloc_j)$, offer the single price
$\Price_j={\CDFB_j}^{(-1)}(1-\XAAlloc_j)$, otherwise randomize between two prices $\Price_j^-$ and
$\Price_j^+$ with probabilities $\theta_j$ and $1-\theta_j$, as explained in \autoref{thm:rev1}. Note that
the randomization must be done for each item independently.
\end{itemize}
\end{definition}
\end{samepage}

\begin{theorem}
The mechanism of \autoref{def:budgetm} obtains at least $1-\frac{1}{e}$ of the revenue of the optimal single
buyer \IP{} mechanism. Furthermore, this mechanism satisfies the requirements of $\gamma$-pre-rounding.
\end{theorem}
\begin{proof}
First, we show that the mechanism obtains at least $1-\frac{1}{e}$ of its benchmark $\Rev(\XAAllocC)$, which
by \autoref{thm:revm} is an upper bound on the optimal revenue. Consider an imaginary replica of the buyer
who has exactly the same valuations as the original buyer, but has a separate budget $\Budget$ for each item.
We call this imaginary buyer the ``super replica''. Furthermore, suppose that any payment received from the
super replica beyond $\Budget$ is lost (i.e., if the super replica pays $\ZRev$, the mechanism receives only
$\min(\ZRev,\Budget)$). Observe that for any assignment of prices, the payment received from the original
buyer and the payment received from the super replica are exactly the same because if the original buyer
has't hit his budget limit then both the original buyer and the super replica will buy the same items and pay
the exact same amount. Otherwise, if the original buyer hits his budget limit, the mechanism receives exactly
$\Budget$ from both the original buyer and the super replica; therefore we only need to show that the revenue
obtained by the mechanism from the super replica is at least $(1-\frac{1}{e})\Rev(\XAAllocC)$. Observe that
from the view point of the super replica there is no connection between different items, so he makes a
decision for each item independently. Let $\ZRev_j$ be the random variable corresponding to the amount paid
by the super replica for item $j$. By \autoref{thm:rev1}, we know that $\Ex{\ZRev_j}=\Rev_j(\XAAllocC_j)$ and
the total revenue received by the mechanism is $\ZRev=\min(\sum_j \ZRev_j, \Budget)$. Notice that
$\ZRev_1,\ldots, \ZRev_\NumItems$ are independent random variables in the range of $\RangeAB{0}{\Budget}$. By
applying \autoref{lem:unc}, we can argue that $\Ex{\min(\sum_j \ZRev_j, \Budget)} \ge (1-\frac{1}{e})
\min(\sum_j \Ex{\ZRev_j}, \Budget) = (1-\frac{1}{e}) \Rev(\XAAllocC)$ which proves our claim.

Next, we show that the mechanism satisfies the requirements of $\gamma$-pre-rounding. Observe that all
$\Rev_j(\wdot)$ are concave, and so is $\Rev(\XAAllocC)$. Furthermore,
$\Rev(\Induced{\XAAllocC}{\ItemSubset})=\min(\sum_{j\in \ItemSubset} \Rev_j(\XAAllocC_j), \Budget)$ is
submodular in $\ItemSubset$ for any $\ItemSubset \subseteq \RangeN{\NumItems}$, and therefore it has a cross
monotonic budget balanced cost share scheme (see \autoref{def:cs}), which completes the proof.
\end{proof}

\begin{lemma}
\label{lem:unc}%
Let $\Budget$ be an arbitrary positive number and let $\ZRev_1, \ldots, \ZRev_\NumItems$ be independent
random variables such that $\ZRev_j \in \RangeAB{0}{\Budget}$, for all $j$. Then the following inequality
holds.
\begin{align*}
    \Ex{\min(\sum_j \ZRev_j, \Budget)} \ge (1-\frac{1}{e^{(\sum_j \Ex{\ZRev_j})/\Budget}}) \Budget \ge (1-\frac{1}{e})\min(\sum_j \Ex{\ZRev_j},\Budget)
\end{align*}
\end{lemma}
\begin{proof}
See \autoref{proof:lem:unc}.
\end{proof}

\subsection{Multi Item (Correlated), Additive, Budget and Matroid Constraints}
\label{sec:corr}%

In this section, we consider a buyer with publicly known budget $\Budget$ who has private correlated additive
valuations for $\NumItems$ items; furthermore, a bundle of items can be allocated to the buyer only if it is
an independent set of a matroid $\Mat=(\RangeN{\NumItems}, \IndepSets)$, where $\Mat$ is publicly known;
equivalently, instead of treating $\Mat$ as a constraint on the allocation, we may assume that the buyer has
matroid valuations, as defined in \autoref{def:matval}. We assume that the buyer has a discrete type space
$\TypeSpace$. Let $\Val_\Type \in \PosReals^\NumItems$ denote the buyer's valuation vector corresponding to
type $\Type \in \TypeSpace$, and let $\PDF(\Type)$ denote its probability. We assume that $\PDF(\wdot)$ is
represented explicitly as a part of the input, i.e., by enumerating all types along with their respective
probabilities. The only private information of the buyer is her type. We present an optimal single buyer
randomized mechanism. This mechanism can be used with $\gamma$-post-rounding (\autoref{def:post}) to yield a
$\gamma_\NumUnits$-approximate multi buyer BIC mechanism. Recall that $\gamma_\NumUnits$ is at least
$\frac{1}{2}$, and approaches $1$ as $\NumUnits\to \infty$, which means the resulting multi buyer mechanism
approaches the optimal multi buyer mechanism as $\NumUnits\to \infty$. Prior to the preliminary version of
this paper, the best approximation for this setting was a $\frac{1}{4}$-approximate BIC mechanism by
\cite{BGGM10}\footnote{The mechanism in \cite{BGGM10} considers demand constraint, which is a special case of
matroid constraints.}. At the time of writing the current version, \cite{HV11} has also presented a
$\frac{1}{2}$-approximate BIC mechanism for the same setting. Note that all of the aforementioned mechanisms
(including the current paper) have running times polynomial only in $\Abs{\TypeSpace}$, which means their
running time may not be polynomial in the input size if $\Abs{\TypeSpace}$ is of exponential size and
$\PDF(\wdot)$ has a compact representation.

Consider the following linear program in which $\XAAlloc_\Type \in \RangeAB{0}{1}^\NumItems$ represents the
marginal allocation probabilities for type $\Type \in \TypeSpace$, and $\Price_\Type$ represents the
corresponding payment. Also let $\Rank{\Mat} : 2^\NumItems \to \RangeMN{0}{\NumItems}$ denote the rank
function of $\Mat$. The optimal value of this LP is obviously an upper bound on the optimal revenue.

\begin{align*}
    \text{maximize} &   &\qquad &  \sum_{\Type \in \TypeSpace} \PDF(\Type) \Price_\Type & \qquad & \tag{$\SRevSC_{corr}$} \label{LP:rev} \\
    \text{subject to}&  &       & \sum_{\Type \in \TypeSpace} \PDF(\Type) \XAAlloc_{\Type j} \le \XAAllocC_j, & & \forall j \in \RangeN{\NumItems} \\
                    &   &       & \sum_{j \in \ItemSubset} \XAAlloc_{\Type j} \le \Rank{\Mat}(\ItemSubset), & & \forall \Type \in \TypeSpace, \forall \ItemSubset \subseteq \RangeN{\NumItems} \\
                    &   &       & \Val_\Type\idot\XAAlloc_\Type - \Price_\Type \ge \Val_{\Type}\idot\XAAlloc_{\Type'} - \Price_{\Type'} , & & \forall \Type,\Type'  \in \TypeSpace  \\
                    &   &       & \XAAlloc_\Type \in \RangeAB{0}{1}^\NumItems, & & \forall \Type \in \TypeSpace \\
                    &   &       & \Price_\Type \in \RangeAB{0}{\Budget}, & & \forall \Type \in \TypeSpace
\end{align*}

Even though the above LP has exponentially many constraints, it can be solved in polynomial time using the
ellipsoid method\footnote{See \cite{S03} for optimization over matroid polytope.}. Next, we present a
mechanism whose expected revenue is equal to the optimal value of the above LP, which also implies that it is
optimal.

\begin{samepage}
\begin{definition}[Mechanism] \ \
\label{def:corr}%
\begin{itemize}
\item
Define the optimal benchmark $\Rev(\XAAllocC)$ to be the optimal value of \eqref{LP:rev} as a function of
$\XAAllocC$.

\item
Given $\XAAllocC$, solve the LP of \eqref{LP:rev} and let $\XAAlloc$ an $\Price$ be an optimal assignment.

\item
Let $\Type$ be the buyer's reported type. Allocate a random subset $\XPAlloc \subseteq \RangeN{\NumItems}$ of
items such that $\XPAlloc$ is an independent set of $\Mat$ and each item $j \in \RangeN{\NumItems}$ is
included in $\XPAlloc $ with a marginal probability of exactly $\XAAlloc_{\Type j}$. This can be archived by
rounding $\XAAlloc_\Type$ to a vertex of the matroid polytope using dependent randomized rounding (see
\cite{CVZ10} and the references therein). Also charge a payment of $\Price_\Type$.
\end{itemize}
\end{definition}
\end{samepage}

\begin{theorem}
\label{thm:corr}%
The mechanism of \autoref{def:corr} is an optimal truthful in expectation revenue maximizing single buyer
mechanism, subject to an upper bound of $\XAAllocC$ on the ex ante allocation rule. Furthermore, it satisfies
all the requirements of the $\gamma$-post-rounding.
\end{theorem}
\begin{proof}
The proof of truthfulness and optimality trivially follows from the linear program of \eqref{LP:rev}. So, we
only focus on proving that this mechanism satisfies the requirements of \autoref{thm:post}. First, observe
that the benchmark function, $\Rev(\XAAllocC)$, is concave (this follows from \autoref{lem:concave}). Second,
observe that the matroid constrains can be interpreted as matroid valuations for the buyer. Third, notice
that the exact ex ante allocation rule can be readily computed from the LP solution, i.e., $\XAAllocA_j =
\sum_\Type \PDF(\Type) \XAAlloc_{\Type j}$ is the exact probability of allocating item $j$. Therefore, the
mechanism satisfies the requirements of $\gamma$-post-rounding.
\end{proof}

\begin{remark}
Observe that by replacing the objective function of \eqref{LP:rev} with $\sum_{\Type \in \TypeSpace}
\PDF(\Type) \Val_{\Type}\idot \XAAlloc_{\Type}$, we get a truthful in expectation welfare maximizing single
buyer mechanism, which can also be used with $\gamma$-post-rounding to obtain a
$\gamma_\NumUnits$-approximate welfare maximizing BIC multiple buyer mechanism.
\end{remark}

\section{Analysis of $\gamma$-Conservative Magician}
\label{sec:magician_analysis}%

In this section, we present the proof of \autoref{thm:magician}. We prove the theorem in two parts. In the
first part, we show that the thresholds computed by the $\gamma$-conservative magician indeed guarantee that
each box is opened with an ex-ante probability at least $\gamma$, assuming that there is enough wand. In the
second part, we show that, for any $\gamma \le 1-\frac{1}{\sqrt{\K+3}}$, the thresholds $\Threshold_i$ are no
more than $\K-1$ for all $i$, which implies that the magician never requires more than $\K$ wands. It can be
shown that a non-adaptive strategy cannot guarantee a probability of more than $1-O(\frac{\sqrt{\ln
\K}}{\sqrt{\K}})$ for opening each box. Furthermore, in \autoref{thm:magician_hardness}, we show that no
algorithm can guarantee a probability of $1-\frac{1}{\sqrt{2\pi \K}}+\epsilon$ or better for opening each
box, for any $\epsilon \ge 0$, which implies that $1-\frac{1}{\sqrt{\K+3}}$ is not far from optimal. Below,
we repeat the dynamic program for computing $\OpenProb_{i}^{\Indw}$, $\Threshold_i$ and
$\BrokenCDF{i}{\wdot}$.

\begin{align}
      \OpenProb_{i}^{\Indw} = \Prx{\OpenRV_{i}=1|\BrokenRV_{i}=\Indw} &=
        \begin{cases}
        1                                                               & \Indw < \Threshold_{i} \\
        (\gamma-\BrokenCDF{i}{\Threshold_{i}-1})/(\BrokenCDF{i}{\Threshold_{i}}-\BrokenCDF{i}{\Threshold_{i}-1}) \quad           & \Indw = \Threshold_{i} \\
        0                                                               & \Indw > \Threshold_{i}
        \end{cases} \tag{$\OpenProb$} \label{eq:dp':s} \\
    \Threshold_{i} &=  \min \{\Indw | \BrokenCDF{i}{\Indw} \ge \gamma \}   \tag{$\Threshold$} \label{eq:dp':theta} \\
    \BrokenCDF{i+1}{\Indw} &=
        \begin{cases}
        \OpenProb_{i}^{\Indw}\XBox_{i}\BrokenCDF{i}{\Indw-1}+(1-\OpenProb_{i}^{\Indw}\XBox_{i})\BrokenCDF{i}{\Indw}     \quad            & i \ge 1,  \Indw \ge 0 \\
        1                                                                                               & i = 0,  \Indw\ge 0 \\
        0                                                                                               & \text{otherwise.}
        \end{cases} \tag{$\BrokenCDFSignW$} \label{eq:dp':cdf}
\end{align}

\paragraph{Part 1}
We show that, assuming there is enough number of wands, the thresholds computed by the dynamic program
guarantee that each box is opened with an ex ante probability of at least $\gamma$.

\begin{enumerate}[(a)]
\item \label{part:1a}%
First we prove that $\Prx{\BrokenRV_i \le \Indw} \ge \BrokenCDF{i}{\Indw}$ by induction on $i$. The base case
is trivial. Suppose the inequality holds for $i \ge 1$, we prove it for $i+1$ as follows.
\begin{align*}
    \Prx{\BrokenRV_{i+1} \le \Indw}  &\ge \Prx{\BrokenRV_{i} \le \Indw-1}+\Prx{\BrokenRV_{i}=\Indw} (1-\OpenProb_{i}^\Indw \XBox_{i}) \\
            &=  \Prx{\BrokenRV_{i} \le \Indw-1}\OpenProb_{i}^\Indw \XBox_{i} +\Prx{\BrokenRV_{i}\le \Indw} (1-\OpenProb_{i}^\Indw \XBox_{i}) \\
            &\ge \BrokenCDF{i}{\Indw-1} \OpenProb_{i}^\Indw \XBox_{i} + \BrokenCDF{i}{\Indw} (1-\OpenProb_{i}^\Indw \XBox_{i})  & \text{by induction hypothesis} \\
            &= \BrokenCDF{i+1}{\Indw} &  \tag{1.a} \label{eq:p1a} \text{by \eqref{eq:dp':cdf}}
\end{align*}
Observe that all of the above inequalities are met with equality if every $\XBox_i$ is the exact probability
of breaking wand for the corresponding box instead of just an upper bound.

\item
Next, we show that each box is opened with probability at least $\gamma$. We shall show that $\Prx{\OpenRV_i=1} \ge
\gamma$.
\begin{align*}
    \Prx{\OpenRV_i=1}   &= \sum_\Indw \Prx{\OpenRV_i=1 | \BrokenRV_{i}=\Indw} \Prx{\BrokenRV_{i}=\Indw} \\
                &= \sum_{\Indw=0}^{\Threshold_i} \OpenProb_i^\Indw \Prx{\BrokenRV_{i}=\Indw} \\
                &= \Prx{\BrokenRV_{i}< \Threshold_i} + \OpenProb_i^{\Threshold_i} \Prx{\BrokenRV_{i}=\Threshold_i} & \text{because $\OpenProb_i^\Indw=1$ for $\Indw < \Threshold_i$} \\
                &= (1-\OpenProb_i^{\Threshold_i}) \Prx{\BrokenRV_{i} < \Threshold_i} + \OpenProb_i^{\Threshold_i} \Prx{\BrokenRV_{i}\le\Threshold_i} \\
                &\ge (1-\OpenProb_i^{\Threshold_i}) \BrokenCDF{i}{\Threshold_i-1} + \OpenProb_i^{\Threshold_i} \BrokenCDF{i}{\Threshold_i} & \text{by \eqref{eq:p1a} }\\
                &= \BrokenCDF{i}{\Threshold_i-1} + \OpenProb_i^{\Threshold_i} (\BrokenCDF{i}{\Threshold_i}-\BrokenCDF{i}{\Threshold_i-1}) \\
                &= \gamma & \text{ by substituting $\OpenProb_i^{\Threshold_i}$ from \eqref{eq:dp':s}}
\end{align*}
Observe that all of the above inequalities are met with equality if each $\XBox_i$ is the exact probability
of breaking a wand for the corresponding box instead of being just an upper bound.
\end{enumerate}

\paragraph{Part 2}
We show that when $\gamma \le 1-\frac{1}{\sqrt{\K+3}}$, the $\gamma$-conservative magician never requires
more than $\K$ wands, i.e., we show that $\Threshold_i \le \K-1$, for all $i$. First, we present an
interpretation of the magician's dynamic program as a stochastic process on an infinite tape with one unit of
infinitely divisible sand.

\begin{definition}[Sand Displacement Process]
\label{def:sand}%
Consider one unit of infinitely divisible sand which is initially at position $0$ on an infinite tape;
positions on the tape are labeled from left to right. The sand is gradually moved to the right and
distributed over the tape in $\N$ rounds. Let $\BrokenCDF{i}{\Indw}$ denote the total amount of sand in
positions $\RangeMN{0}{\Indw}$ at the beginning of round $i \in \RangeN{\N}$. At each round $i$ the following
happens.
\begin{enumerate}[(I)]
\item
The leftmost $\gamma$-fraction of the sand is selected by identifying the smallest threshold $\Threshold_i$
such that $\BrokenCDF{i}{\Threshold_i} \ge \gamma$ and then selecting all the sand in positions
$\RangeMN{0}{\Threshold_i-1}$ and selecting a fraction of the sand at position $\Threshold_i$ itself such
that the total amount of selected is equal to $\gamma$. Formally, if $\BrokenSel{i}{\Indw}$ denotes the total
amount of sand selected from $\RangeMN{0}{\Indw}$, the selection of sand is such that
$\BrokenSel{i}{\Indw}=\min(\BrokenCDF{i}{\Indw}, \gamma)$, for every $\Indw$. In particular, this implies
that only a fraction of the sand at position $\Threshold_i$ itself might be selected, however all the sand to
the left of position $\Threshold_i$ is selected.

\item
The selected sand is moved one position to the right as follows. Simultaneously, for all $\Indw \in
\RangeMN{0}{\Threshold_i}$, an $\XBox_i$ fraction of the sand selected from position $\Indw$ is moved to
position $\Indw+1$. For every $\Indw$, the total amount of sand in positions $\RangeMN{0}{\Indw}$ at the end
of round $i$ is given by the following equation.
\begin{align}
    \BrokenCDF{i+1}{\Indw} &= \BrokenCDF{i}{\Indw}-\XBox_{i}\BrokenSel{i}{\Indw}+ \XBox_{i}  \BrokenSel{i}{\Indw-1} \tag{$\BrokenCDFSign_{eq}$} \label{eq:phi}
\end{align}
\end{enumerate}
\end{definition}

It is easy to see that $\Threshold_i$ and $\BrokenCDF{i}{\Indw}$ resulting from the above process are exactly
the same as those computed by the $\gamma$-conservative magician.

Consider a conceptual barrier which is at position $\Threshold_i+1$ at the beginning of round $i$ and is
moved to position $\Threshold_{i+1}+1$ for the next round, for each $i \in \RangeN{\N}$. It is easy to verify
(i.e., by induction) that the sand never crosses to the right side of the barrier (i.e.,
$\BrokenCDF{i+1}{\Threshold_i+1}=1$). The following theorem implies that the sand remains concentrated near
the barrier throughout the process.

\begin{theorem}[Sand]
\label{thm:sand}%
Throughout the sand displacement process (\autoref{def:sand}), at the beginning of each round $i \in
\RangeN{\N}$, the following inequality holds.
\begin{align}
    \BrokenCDF{i}{\Indw} &< \gamma \BrokenCDF{i}{\Indw+1} & & \forall i \in \RangeN{\N}, \forall \Indw \in \RangeN{\Threshold_{i}} \tag{$\BrokenCDFSign_{ineq}$} \label{eq:phi_ineq}
\end{align}
Furthermore, the average distance of the sand from the barrier, denoted by $\Distance_i$, is upper bounded by
the following inequalities.
\begin{align}
    \Distance_i &\le \frac{1-\gamma^{\Threshold_i+1}}{1-\gamma} < \frac{1}{1-\gamma}  & & \forall i \in \RangeN{\N} \tag{$\Distance$} \label{eq:barrier_dist}
\end{align}
The first inequality is strict except for $i=1$.
\end{theorem}

\begin{proof}
We start by proving inequality \eqref{eq:phi_ineq}, by induction on $i$. The base case of $i=1$ is trivial
because all the sand is at position $0$ and $\Threshold_1 = 0$. Suppose the inequality holds at the beginning
of round $i$ and for all $\Indw \in \RangeN{\Threshold_{i}}$, we show that it holds at the beginning of round
$i+1$ and for all $\Indw \in \RangeN{\Threshold_{i+1}}$. Consider any $\Indw \in \RangeN{\Threshold_{i+1}}$.
Note that $\Threshold_{i} \le \Threshold_{i+1} \le \Threshold_{i}+1$, so there are two possible cases:

\begin{itemize}
\item \About{Case 1}
If $\Indw \in \RangeN{\Threshold_{i}}$, then
\begin{align*}
    \BrokenCDF{i+1}{\Indw-1} &= \BrokenCDF{i}{\Indw-1}-\XBox_i  \BrokenSel{i}{\Indw-1}+\XBox_i  \BrokenSel{i}{\Indw-2} && \text{by \eqref{eq:phi}.} \\
                &= (1-\XBox_i)\BrokenCDF{i}{\Indw-1} + \XBox_i \BrokenCDF{i}{\Indw-2} && \text{by $\BrokenSel{i}{\Indw'}=\BrokenCDF{i}{\Indw'}$, for $\scriptstyle \Indw' < \Threshold_{i}$.} \\
                &< (1-\XBox_i)\gamma\BrokenCDF{i}{\Indw} + \XBox_i \gamma \BrokenCDF{i}{\Indw-1} && \text{by induction hypothesis.} \\
                &= \gamma\sdot\left(\BrokenCDF{i}{\Indw}-\XBox_i \BrokenCDF{i}{\Indw}+ \XBox_i \BrokenSel{i}{\Indw-1}\right) && \text{by $\BrokenSel{i}{\Indw'}=\BrokenCDF{i}{\Indw'}$, for $\scriptstyle \Indw' < \Threshold_{i}$.} \\
                &\le \gamma\sdot\left(\BrokenCDF{i}{\Indw}-\XBox_i \BrokenSel{i}{\Indw}+ \XBox_i \BrokenSel{i}{\Indw-1}\right) && \text{by $\BrokenSel{i}{\Indw'}\le \BrokenCDF{i}{\Indw'}$, for all $\Indw'$.} \\
                &= \gamma\BrokenCDF{i+1}{\Indw} && \text{by \eqref{eq:phi}.}
\end{align*}

\item \About{Case 2}
If $\Indw = \Threshold_{i+1}=\Threshold_{i}+1$, then by definition of $\Threshold_{i+1}$, it must be
$\BrokenCDF{i+1}{\Threshold_{i+1}-1} < \gamma$. Furthermore, at the end of round $i$, all the sand must be in
the range $\RangeMN{0}{\Threshold_{i}+1}$, so $\BrokenCDF{i+1}{\Threshold_{i+1}} =
\BrokenCDF{i+1}{\Threshold_{i}+1} = 1$. Consequently, we can conclude that
$\BrokenCDF{i+1}{\Threshold_{i+1}-1} \le \gamma  \BrokenCDF{i+1}{\Threshold_{i+1}}$ which proves the claim.
\end{itemize}

Next, we prove that the average distance of the sand from the barrier, $\Distance_i$, is no more than
$\frac{1-\gamma^{\Threshold_i +1}}{1-\gamma}$ at the beginning of round $i$.
\begin{align*}
    \Distance_i    &= \sum_{\Indw=0}^{\Threshold_i} \BrokenCDF{i}{\Indw} & \text{sand at position $\Indw$ is counted exactly $\Threshold_i-\Indw+1$ times} \\
            &\le \sum_{\Indw=0}^{\Threshold_i} \gamma^{\Threshold_i-\Indw} \BrokenCDF{i}{\Threshold_i} & \text{by \eqref{eq:phi}} \\
            &\le \sum_{\Indw=0}^{\Threshold_i} \gamma^{\Threshold_i-\Indw} & \text{because $\BrokenCDF{i}{\Threshold_i} \le 1$} \\
            &= \frac{1-\gamma^{\Threshold_i+1}}{1-\gamma}
\end{align*}
Note that the second inequality is strict unless $\Threshold_i=0$; furthermore conditioned on
$\Threshold_i=0$, the third inequality is strict except for round $i=1$; therefore, the upper-bound is strict
except for the first round.
\end{proof}

\begin{theorem}[Barrier]
\label{thm:barrier}%
Consider the process of \autoref{def:sand}. If $\gamma \le 1-\frac{1}{\sqrt{\K+3}}$ and $\sum_{i=1}^\N
\XBox_i \le \K$ for some $\K \in \PosIntsZ$, then the barrier never gets past position $\K$, i.e.,
$\Threshold_i \le \K-1$ for all $i \in \RangeN{\N}$.
\end{theorem}
\begin{proof}
We shall show that the barrier never gets past position $\K$ throughout the process, which proves the theorem
(recall that the barrier is defined to be at position $\Threshold_i+1$). Let $\Distance'_i$ denote the
average distance of the sand from the origin, and let $\Distance_i$ denote the average distance of the sand
from the barrier, at the beginning of round $i$. Observe that $\Distance'_i+\Distance_i = \Threshold_i+1$;
furthermore, $\Distance'_i = \Distance'_{i-1}+\gamma \XBox_{i-1}=\gamma \sum_{r=1}^{i-1} \XBox_r$, i.e., the
average distance of the sand from the origin is increased exactly by $\gamma \XBox_{i-1}$ during round $i-1$
(because the amount of selected sand is exactly $\gamma$ and $\XBox_{i-1}$ fraction of the selected sand is
moved one position to the right). By applying \autoref{thm:sand}, we get the following inequality.
\begin{align*}
    \Threshold_i+1  &= \Distance'_i + \Distance_i < \gamma \sum_{r=1}^{i-1} \XBox_r + \frac{1-\gamma^{\Threshold_i+1}}{1-\gamma}
        \le \gamma \K + \frac{1-\gamma^{\Threshold_i+1}}{1-\gamma}
\end{align*}

In order to show that the barrier never gets past position $\K$, it is enough to show that the above
inequality cannot hold for $\Threshold_i \ge \K$; in fact it is just enough to show that it cannot hold for
$\Threshold_i=\K$~\footnote{Because in order for the barrier to get past position $\K$, it must first fall on
position $\K+1$.}; alternatively, it is enough to show that the complement of the above inequality holds for
$\Threshold_i = \K$.

\begin{alignat*}{2}
    \K+1   &\ge \gamma \K + \frac{1-\gamma^{\K+1}}{1-\gamma} \tag{$\Gamma$} \label{eq:Gamma}
\end{alignat*}

Consider the the stronger inequality $\K+1 \ge \gamma \K  + \frac{1}{1-\gamma}$ which is quadratic in
$\gamma$ and yields a bound of $\gamma \le 1-\frac{1}{1/2+\sqrt{\K+1/4}}$; this bound is in fact imposes a
looser constraint than $\gamma \le 1-\frac{1}{\sqrt{\K+3}}$ when $\K \ge 7$. Furthermore it can be verified
(by direct calculation) that inequality \eqref{eq:Gamma} holds, for $\K < 7$ and $\gamma \le
1-\frac{1}{\sqrt{\K+3}}$. That completes the proof.
\end{proof}

\autoref{thm:barrier} implies that a $\gamma$-conservative magician requires no more than $\K$ wands,
assuming that $\gamma \le 1-\frac{1}{\sqrt{\K+3}}$. That completes the proof of \autoref{thm:magician}.

\section{Multi Unit Demands}
\label{sec:multi_unit}%

In this section, we show that the more general model, in which each buyer may need more than one unit but no
more than $\frac{1}{\MaxDemandFrac}$ of all units of each item, can be reduced to the simpler model in which
there are at least $\NumUnits$ units of every item and no buyer demands more than $1$ unit of each item.

\begin{definition}[Multi Unit Demand Market Transformation]
\label{def:trans}%
Let $\NumUnits_j$ denote the number of units of item $j$. Define
$\BinCount_j=\Floor{\frac{\NumUnits_j}{\MaxDemandFrac}} $ and divide the units of item $j$ almost equally
into $\BinCount_j$ bins (i.e., each bin will contain either $\BinCount_j$ or $\BinCount_j+1$ units). Create a
new item type for each bin (i.e., units from the same bin has the same type, but units from different bins
are treated as different types of item).
\end{definition}

\begin{theorem}
\label{thm:trans}%
Let $\Mechs$ be the space of feasible mechanisms, in the original (multi unit demand) market, which do not
allocate more than $\frac{1}{\MaxDemandFrac}$ of all units of each item to any single buyer. Similarly, let
$\Mechs^{(1)}$ be the space of feasible mechanisms, in the transformed market, which do not allocate more
than one unit of each item to any single buyer. Any mechanism in $\Mechs$ can be interpreted as a mechanism
in $\Mechs^{(1)}$ and vice-versa with the same allocations/payments. Therefore, in order to find the optimal
mechanism in the original market, it is enough to find the optimal mechanism in the transformed market.
\end{theorem}
\begin{proof}
First, we show that any mechanism in $\Mech \in \Mechs^{(1)}$ can be interpreted as a mechanism in $\Mechs$.
That is trivially true because $\Mech$ allocates to each buyer at most one unit from each bin, which is at
most $\BinCount_j$ units of each item $j$ of the original market, which is no more than
$\frac{1}{\MaxDemandFrac}$ of all units of item $j$.

Next, we show that any mechanism $\Mech \in \Mechs$ can be interpreted as a mechanism in $\Mechs^{(1)}$. For
every $j$, we create a list $\List_j$ of all the bins of item $j$. $\List_j$ is initially sorted in
decreasing order of the size of the bins.  Let $\XPAlloc^\Mech_{ij}$ be the number of units of item $j$
allocated to buyer $i$ by $\Mech$. We specify the allocations in the transformed market as follows. For each
buyer $i$ we repeat the following, $\XPAlloc^\Mech_{ij}$ times: Allocate one unit from the bin that is first
in the list $\List_j$ and then move the bin back to the end of the list. It is easy to see that no two units
from the same bin are allocated to the same buyer, which completes the proof.
\end{proof}

Note that by \autoref{thm:trans}, any mechanism in the original market is equivalent to a mechanism in the
transformed market with the exact same allocations/payments from the perspective of buyers. Therefore, WLOG,
we can work with the transformed market and only consider mechanisms in this market. However, to use our
generic multi buyer mechanisms in the transformed market, the underlying single buyer mechanisms should be
capable of handling correlated valuations, because units of the same item, even when labeled with different
types, are perfect substitutes from the view point of a buyer. Among the single buyer mechanisms presented in
this paper, only the mechanism explained in \SRef{sec:corr} can handle correlated valuations.

\section{Conclusion}

In this paper, for Bayesian combinatorial auctions, we presented an approximate reduction from multi buyer
problem to single buyer problems, which leads to the following conclusions.

\begin{itemize}
\item \About{Market size}
As the ratio of the maximum demand to supply (i.e., $\frac{1}{\MaxDemandFrac}$) decreases, less coordination
is required on decisions made for different buyers; i.e., as $\frac{1}{\MaxDemandFrac} \to 0$, the optimal
mechanism treats each buyer almost independently of other buyers. Observe that all of the approximation
factors in this paper only depend on $\MaxDemandFrac$ (i.e., $\gamma_\MaxDemandFrac$) and not on
$\NumAgents$. It suggests that, for characterizing asymptotic properties of such markets, the right parameter
to consider is perhaps the ratio of the maximum demand to supply; in particular, notice that the number of
buyers is irrelevant.

\item \About{Computational hardness}
For mechanism design problems in a variety of settings, the difficulty of making coordinated optimal
decisions for multiple buyers can be avoided by losing a small constant factor in the objective (i.e., losing
only a $\frac{1}{\sqrt{\MaxDemandFrac+3}}$ fraction of the objective), therefore the main difficulty of
constructing constant factor approximation mechanisms in multi dimensional settings stems from the difficulty
of designing single buyer mechanisms, which ultimately stems from the IC constraints in the single buyer
problem.
\end{itemize}

\section{Acknowledgment}
I thank Jason Hartline for many helpful suggestions and comments. I also thank Azarakhsh Malekian for a fruitful discussion that
lead to the proof of the sand theorem (\autoref{thm:sand}).

\bibliography{bibs}

\appendix

\section{Missing Proofs}
\label{sec:proofs}%

\phantomsection
\begin{proof}[Proof of \autoref{thm:magician_hardness}]
\label{proof:thm:magician_hardness}%
Suppose we create $\N$ boxes and in each box, independently, we put $\$1$ with probability $\frac{\K}{\N}$.
If the magician opens a box containing a $\$1$, then he gets the $\$1$ but we break his wand (i.e.,
$\XBox_i=\frac{\K}{\N}$). Observe that the expected total prize is $\K$ dollars, but because we put a dollar
in each box independently, there are some instances in which there are more than $\K$ non-empty boxes but the
magician cannot win more than $\K$ dollars at any instance. Let $\XBoxRV_i$ be the indicator random variable
which is $1$ if{f} there is a dollar in box $i$. The expected total prize is $\Ex{\sum_i \XBoxRV_i}=\K$, but
the expected prize that the magician can win is at most $\Ex{\min(\sum_i \XBoxRV_i, \K)}$. It can be verified
that $\Ex{\min(\sum_i \XBoxRV_i, \K)} \approx (1-\frac{\K^\K}{e^\K \K!}) \K$ asymptotically as $\N \to
\infty$. In fact, for any positive $\epsilon$, there is a large enough $\N$ such that $\Ex{\min(\sum_i
\XBoxRV_i, \K)} < (1-\frac{\K^\K}{e^\K \K!}+\epsilon) \K$. On the other hand, if a magician can guarantee
that every box is opened with probability at least $\gamma=1-\frac{\K^\K}{e^\K \K!}+\epsilon$, then he will
be able to obtain a prize of at least $\sum_i \gamma \Ex{\XBoxRV_i} = (1-\frac{\K^\K}{e^\K \K!}+\epsilon) \K$
in expectation which is a contradiction; therefore it is not possible to make such a guarantee.
\end{proof}

\phantomsection
\begin{proof}[Proof of \autoref{thm:pre}]
\label{proof:thm:pre}%
First, we show that each $\ItemSubset_i$ includes each item $j$ with probability at least $\gamma$. Observe
that for each item $j$, a sequence of $\NumAgents$ boxes are presented to the $j^{th}$ magician with
probabilities $\XAAllocC_{1j}, \ldots, \XAAllocC_{\NumAgents j}$ written on them. Since $\sum_i
\XAAllocC_{ij} \le \NumUnits_j$ and $\gamma \in \RangeAB{0}{\gamma_\NumUnits}$, we can argue that each box is
opened with probability at least $\gamma$ (see \autoref{thm:magician} and \autoref{def:magician_param});
therefore $\ItemSubset_i$ includes each item $j$ with probability at least $\gamma$.

Next, we show that the expected objective value of $\gamma$-pre-rounding is at least $\gamma \alpha$-fraction
of the expected objective value of the optimal mechanism in $\Mechs$. Note that by
\autoref{thm:decomposition}, the expected objective value of the optimal mechanism in $\Mechs$ is upper
bounded by the optimal value of \eqref{P:OPT'} which is $\sum_i \Rev_i(\XAAllocC_i)$; therefore, it is enough
to show that $\Ex[\ItemSubset_i]{\Rev_i(\Induced{\XAAllocC_i}{\ItemSubset_i})} \ge \gamma \alpha
\Rev_i(\XAAllocC_i)$, i.e., the expected objective value that
$\Cons{\Mech_i}{\Induced{\XAAllocC_i}{\ItemSubset_i}}$ obtains from buyer $i$ is at least $\gamma \alpha
\Rev_i(\XAAllocC_i)$. Let $\xi_i$ be a budget balanced cross monotonic cost share function for
$\Rev_i(\wdot)$; then

\begin{align*}
   \Ex[\ItemSubset_i]{\Rev_i(\Induced{\XAAllocC_i}{\ItemSubset_i})}
        &= \Ex[\ItemSubset_i]{\sum_{j \in \ItemSubset_i} \xi_i(j, \Induced{\XAAllocC_i}{\ItemSubset_i})} & &\text{because $\xi_i$ is budget balanced} \\
        &\ge \Ex[\ItemSubset_i]{\sum_{j \in \ItemSubset_i} \xi_i(j, \XAAllocC_i[\{1,\ldots, \NumItems\}])} & &\text{because $\xi_i$ is cross monotonic} \\
        &= \sum_{j \in \RangeN{\NumItems}} \Prx{j \in \ItemSubset_i} \xi_i(j, \XAAllocC_i) \\
        &\ge \sum_{j \in \RangeN{\NumItems}} \gamma \xi_i(j, \XAAllocC_i) \\
        &= \gamma \Rev_i(\XAAllocC_i) & &\text{because $\xi_i$ is budget balanced}
\end{align*}

Next, we show that the multi buyer mechanism based on $\gamma$-pre-rounding is in $\Mechs$ and it is dominant
strategy incentive compatible (DSIC). The fact that this mechanism is in $\Mechs$ follows from assumption
\myref{A}{asm:decompose} and the fact that for each item $j$, the corresponding magician breaks no more than
$\NumUnits_j$ wands, which means no more than $\NumUnits_j$ units are allocated at any instance. To show that
it is DSIC, observe that the only way the reports of other buyers could affect the outcome of buyer $i$ is by
affecting $\ItemSubset_i$, yet $\Cons{\Mech_i}{\Induced{\XAAllocC_i}{\ItemSubset_i}}$ is a mechanism in
$\Mechs_i$, so it is incentive compatible mechanism for any choice of $\ItemSubset_i$; therefore the
resulting mechanism is DSIC. Observe that this mechanism also preserves all of the ex post properties of each
$\Mech_i$ (e.g., individual rationality).
\end{proof}

\phantomsection
\begin{proof}[Proof of \autoref{thm:post}]
\label{proof:thm:post}%

First, we show that each $\ItemSubset_i$ includes each item $j$ with probability exactly $\gamma$. Observe
that for each item $j$, a sequence of $\NumAgents$ boxes are presented to the $j^{th}$ magician with
probabilities $\XAAllocA_{1j}, \ldots, \XAAllocA_{\NumAgents j}$ written on them. Since $\gamma \in
\RangeAB{0}{\gamma_\NumUnits}$ and $\sum_i \XAAllocA_{ij} \le \NumUnits_j$ and because each
$\Cons{\Mech_i}{\XAAllocC_i}$ allocates each item $j$ with probability exactly $\XAAllocA_{ij}$, we can argue
that each box is opened with probability exactly $\gamma$ (see \autoref{thm:magician} and
\autoref{def:magician_param}); therefore $\ItemSubset_i$ includes each item $j$ with probability exactly
$\gamma$.

Next, we show that $\gamma$-post-rounding obtains in expectation at least $\gamma\alpha$-fraction of the
expected objective value of the optimal mechanism in $\Mechs$. Note that by \autoref{thm:decomposition} the
expected objective value of the optimal mechanism in $\Mechs$ is upper bounded by the optimal value of
\eqref{P:OPT'} which is $\sum_i \Rev_i(\XAAllocC_i)$; therefore, it is enough to show that
$\Ex[\Type_i,\XPAlloc_i,\Pay_i]{\Obj_i(\Type_i, \XPAlloc_i, \Pay_i)} \ge \gamma \alpha \Rev_i(\XAAllocC_i)$,
i.e., the expected objective value that $\gamma$-post-rounding obtains from buyer $i$ is at least $\gamma
\alpha \Rev_i(\XAAllocC_i)$. Let $\xi_i$ be a budget balanced cross monotonic cost share function for
$\Obj_i$ as required by \myref{A'}{asm:post:obj}; then

\begin{align*}
    \Ex[\Type_i,\XPAlloc_i,\Pay_i]{\Obj_i(\Type_i, \XPAlloc_i, \Pay_i)}
        &= \Ex[\Type_i,\XPAlloc_i,\Pay_i]{\Obj_i(\Type_i, \XPAlloc_i, 0)+ \PayFactor_i \Pay_i} & &\text{By \myref{A'}{asm:post:obj}}\\
        &= \Ex[\Type_i,\XPAlloc_i,\Pay_i]{\sum_{j \in \XPAlloc_i} \xi_i(j, \Type_i, \XPAlloc_i) + \PayFactor_i \Pay_i} & &\text{because $\xi_i$ is budget balanced}\\
        &\ge \Ex[\Type_i,\XPAlloc_i,\Pay_i]{\sum_{j \in \XPAlloc_i} \xi_i(j, \Type_i, \XPAlloc'_i) + \PayFactor_i \Pay_i} & &\text{because $\xi_i$ is cross monotonic} \\
        &= \Ex[\Type_i,\XPAlloc'_i,\Pay'_i, \ItemSubset_i]{\sum_{j \in \XPAlloc'_i} \Prx{j \in \ItemSubset_i} \xi_i(j, \Type_i, \XPAlloc'_i) + \PayFactor_i \gamma \Pay'_i} \\
        &= \Ex[\Type_i,\XPAlloc'_i,\Pay'_i]{\sum_{j \in \XPAlloc'_i} \gamma \xi_i(j, \Type_i, \XPAlloc'_i) + \PayFactor_i \gamma \Pay'_i} \\
        &= \gamma \Ex[\Type_i,\XPAlloc'_i,\Pay'_i]{\Obj_i(\Type_i, \XPAlloc'_i, \Pay'_i)} \\
        &\ge \gamma \alpha \Rev_i(\XAAllocC_i) \\
\end{align*}
Note that the last step follows because $\Ex[\Type_i,\XPAlloc'_i,\Pay'_i]{\Obj_i(\Type_i, \XPAlloc'_i,
\Pay'_i)}$ is exactly the expected objective value of $\Cons{\Mech_i}{\XAAllocC_i}$ which is at least $\alpha
\Rev_i(\XAAllocC_i)$.

Next, we show that $\gamma$-post-rounding is Bayesian incentive compatible (BIC) and does not over allocate
any item. Consider any arbitrary buyer $i$. Observe that each item $j \in \XPAlloc'_i$ is included in
$\XPAlloc_i$ with a probability of exactly $\gamma$; furthermore, by \myref{A'}{asm:post:valuations},
valuations of buyer $i$ can be interpreted as a weighted rank function of some matroid; WLOG, we may assume
that $\XPAlloc'_i$ is always an independent set of this matroid\footnote{Otherwise, we could replace
$\XPAlloc'_i$ by a maximum weight independent subset of $\XPAlloc'_i$.}; therefore, the valuation of the
buyer for the items in $\XPAlloc'_i$ is additive; consequently, her expected valuation for $\XPAlloc_i$ is
exactly $\gamma$ times her valuation for $\XPAlloc'_i$. Observe that both the expected valuation and the
expected payment of buyer $i$ are scaled by $\gamma$ for any outcome of $\Cons{\Mech_i}{\XAAllocC_i}$ and
$\Cons{\Mech_i}{\XAAllocC_i}$ itself was incentive compatible; therefore, the resulting mechanism is
incentive also incentive compatible. However, the final mechanism is only Bayesian incentive compatible
because $\ItemSubset_i$ depends on the typers/reports of buyers other that $i$ \footnote{I.e., $\Prx{j \in
\ItemSubset_i}$ is equal to $\gamma$ only in expectation over other buyers' reports}. Also note that the
mechanism does not over allocate any item, because for each unit of item $j$ being allocated one of the
$\NumUnits_j$ wands of the $j^{th}$ magician breaks.
\end{proof}

\phantomsection
\begin{proof}[Proof of \autoref{lem:concave}]
\label{proof:lem:concave}%
The proof is very similar to the proof of \autoref{thm:opt_concave}. To show that $\Rev(\XAAllocC)$ is
concave, it is enough to show that for any $\XAAllocC$ and $\XAAllocC'$ and any $\beta \in \RangeAB{0}{1}$,
$\Rev(\beta \XAAllocC+(1-\beta) \XAAllocC') \ge \beta \Rev(\XAAllocC)+(1-\beta)\Rev(\XAAllocC')$. Let $y$ and
$y'$ be the optimal assignments for the convex program subject to $\XAAllocC$ and $\XAAllocC'$ respectively;
then $y''= \beta y + (1-\beta)y'$ is also a feasible assignment for the convex program subject to $\beta
\XAAllocC+(1-\beta) \XAAllocC'$; therefore, $\Rev(\beta \XAAllocC+(1-\beta) \XAAllocC')$ must be at least
$\SObj(\beta y+(1-\beta) y')$; on the other hand $\SObj(\wdot)$ is concave, so $\SObj(\beta y+(1-\beta) y')
\ge \beta \SObj(y)+(1-\beta)\SObj(y') = \beta \Rev(\XAAllocC)+(1-\beta)\Rev(\XAAllocC')$. That proves the
claim.
\end{proof}

 \phantomsection
\begin{proof}[Proof of \autoref{lem:unc}]
\label{proof:lem:unc}%
Let $\mu=\sum_j \Ex{\ZRev_j}$. Define the random variables $\YRev_j = \max(\YRev_{j-1}-\ZRev_j, 0)$ and
$\YRev_0 = \Budget$. Observe that for each $j$, $\YRev_j = \max(\Budget-\sum_{r=1}^j \ZRev_r, 0)$, so
$\min(\sum_{r=1}^j \ZRev_r, \Budget)+\YRev_j= \Budget$. Therefore $\Ex{\min(\sum_{r=1}^j \ZRev_r,
\Budget)}+\Ex{\YRev_j}=\Budget$ and to prove the theorem it is enough to show that $\Ex{\YRev_\NumItems} \le
\frac{1}{e^{\mu/\Budget}} \Budget$. We first show that
\begin{align}
    \YRev_j & \le (1-\frac{\Ex{\ZRev_j}}{\Budget}) \YRev_{j-1}. \label{eq:Wj} \tag{$\YRev_j$}
\intertext{Consequently}
    \YRev_\NumItems  & \le \Budget  \prod_{j=1}^\NumItems (1-\frac{\Ex{\ZRev_j}}{\Budget}) \\
                    & \le \Budget  \frac{1}{e^{\mu/\Budget}}
\end{align}
The last inequality follows because $\prod_{j=1}^\NumItems (1-\frac{\Ex{\ZRev_j}}{\Budget})$ takes its
maximum when $\frac{\Ex{\ZRev_j}}{\Budget} = \frac{\mu}{\NumItems \Budget}$ (for all $j$) and $\NumItems \to
\infty$.

To prove the second inequality in the statement of the lemma we can use the fact that $(1-x^a) \ge (1-x)a$
for any $a \le 1$, and conclude that $(1-\frac{1}{e^{\mu/\Budget}})\Budget \ge
(1-\frac{1}{e^{\min(\mu,\Budget)/\Budget}})\Budget \ge (1-\frac{1}{e}) \frac{\min(\mu,\Budget)}{\Budget}
\Budget = (1-\frac{1}{e}) \min(\mu,\Budget)$.

To complete the proof, we prove inequality \eqref{eq:Wj} as follows.
\begin{align*}
    \Ex{\YRev_j}   & = \Ex{\max(\YRev_{j-1}-\ZRev_j, 0)} \\
                    & \le \Ex{\max(\YRev_{j-1}-\ZRev_j \frac{\YRev_{j-1}}{\Budget}, 0)} && \text{because $\frac{\YRev_{j-1}}{\Budget} \le 1$ } \\
                    & = \Ex{\YRev_{j-1}-\ZRev_j \frac{\YRev_{j-1}}{\Budget}} && \text{because $\frac{\ZRev_j}{\Budget} \le 1$ } \\
                    & = \Ex{\YRev_{j-1}}-\frac{1}{\Budget}\Ex{\ZRev_j \YRev_{j-1}} \\
                    & \le \Ex{\YRev_{j-1}}-\frac{1}{\Budget}\Ex{\ZRev_j} \Ex{\YRev_{j-1}} && \text{because $\ZRev_j$ and $\YRev_{j-1}$ are independent.} \\
                    & = (1-\frac{\Ex{\ZRev_j}}{\Budget}) \Ex{\YRev_{j-1}}
\end{align*}
\end{proof}

\phantomsection
\begin{proof}[Proof of \autoref{lem:dp}]
\label{proof:lem:dp}%
To prove the claim, it is enough to show that $\frac{\TailRev_1}{\sum_j \XAAlloc_j \Price_j} \ge
\frac{1}{2}$. WLOG, we may assume that $\sum_j \Price_j \XAAlloc_j = 1$ since we can scale
$\Price_1,\ldots,\Price_\NumItems$ by a constant $c=\frac{1}{\sum_j \XAAlloc_j \Price_j}$ and this will also
scale $\TailRev_1,\ldots, \TailRev_\NumItems$ by the same constant $c$, so their ratio is not be affected.
Consider the following LP and observe that $\XAAlloc_j$, $\Price_j$, and $\TailRev_j$, as defined in the
statement of the lemma, form a feasible assignment for this LP. If we show that the optimal objective value
of the LP is bounded below by $\frac{1}{2}$, any feasible assignment yields an objective value of at least
$\frac{1}{2}$, and therefore $\frac{\TailRev_1}{\sum_j \XAAlloc_j \Price_j} \ge \frac{1}{2}$ which proves the
lemma. In the following LP, $\Price_j$ and $\TailRev_j$ are variables and everything else is constant.

\begin{align*}
    \text{minimize} &   &\qquad & \TailRev_1 & \qquad & \notag \\
    \text{subject to}&  &       \TailRev_j &\ge \XAAlloc_j \Price_j + (1-\XAAlloc_j) \TailRev_{j+1}, &  & \forall j \in \RangeN{\NumItems} \tag{$\alpha_j$}\label{eq:alpha_j} \\
                    &   &       \TailRev_j &\ge \TailRev_{j+1}, & &  \forall j \in \RangeN{\NumItems} \tag{$\beta_j$}\label{eq:beta_j} \\
                    &   &       \sum_{j=1}^\NumItems \XAAlloc_j \Price_j &\ge 1 \tag{$\gamma$} \label{eq:gamma} \\
                    &   &       \Price_j &\ge 0, & & \forall j \in \RangeN{\NumItems} \\
                    &   &       \TailRev_j &\ge 0, & & \forall j \in \RangeN{\NumItems+1}
\end{align*}

To prove that the optimal value of the above LP is bounded below by $\frac{1}{2}$, we construct a feasible
assignment for its dual LP, obtaining a value of $\frac{1}{2}$. The dual LP is as follows.

\begin{align*}
    \text{maximize} &   &\qquad & \gamma & \qquad & \notag \\
    \text{subject to}&  &      \gamma &\le \alpha_j, &  & \forall j \in \RangeN{\NumItems} \tag{$\Price_j$} \\
                    &   &       \alpha_1 + \beta_1 &\le 1 \tag{$\TailRev_1$} \\
                    &   &       \alpha_j + \beta_j &\le (1-\XAAlloc_{j-1}) \alpha_{j-1} + \beta_{j-1}, & & \forall j \in \RangeMN{2}{\NumItems} \tag{$\TailRev_j$} \\
                    &   &       0 &\le (1-\XAAlloc_{\NumItems}) \alpha_{\NumItems} + \beta_{\NumItems} \tag{$\TailRev_{\NumItems+1}$} \\
                    &   &       \alpha_j &\ge 0, \beta_j \ge 0, \gamma \ge 0, & & \forall j \in \RangeN{\NumItems}
\end{align*}

We construct an assignment for the dual LP as follows. Set $\alpha_j = \gamma$ and set $\beta_j =
\beta_{j-1}-\XAAlloc_{j-1} \gamma$ for all $j$, except that for $j=1$ we set $\beta_1 = 1-\gamma$. From this
assignment we get $\beta_j = 1-\gamma - \gamma \sum_{\ell=1}^{j-1}\XAAlloc_\ell$. Observe that we get a
feasible assignment as long as all $\beta_j$ resulting from this assignment are non-negative. Furthermore, it
is easy to see that $\beta_j \ge 1 - \gamma -\gamma \sum_{\ell=1}^\NumItems \XAAlloc_\ell \ge 1 - 2\gamma$
because $\sum_j \XAAlloc_j \le 1$. Therefore, by setting $\gamma=\frac{1}{2}$, all $\beta_j$ are non-negative
and we always get a feasible assignment for the dual LP with an objective value of $\frac{1}{2}$, which
completes the proof.
\end{proof}

\end{document}